\documentclass[a4paper,UKenglish,cleveref, autoref, thm-restate]{lipics-v2021}

\nolinenumbers

\usepackage[utf8]{inputenc}
\usepackage[T1]{fontenc}
\usepackage{xspace}
\usepackage{graphicx}
\usepackage{tikz}
\usepackage{multirow}
\usepackage{floatrow}
\usepackage{enumerate}
\usepackage{hyperref}
\hypersetup{
  pdftitle = {},
  pdfauthor=  {},
  colorlinks = true,
  linkcolor = black!30!red,
  citecolor = black!30!green
}
\usepackage{amsmath, amsthm, amssymb}
\usepackage{thmtools}
\usepackage{thm-restate}
\usepackage{todonotes}

\newcommand{\RR}{\mathbb{R}}
\newcommand{\NN}{\mathbb{N}}

\newcommand{\G}{\mathcal{G}}
\newcommand{\D}{\mathcal{D}} 
\newcommand{\R}{\mathcal{R}} 
\newcommand{\C}{\mathcal{C}} 
\newcommand{\E}{\mathcal{E}} 
 
\newcommand{\T}{\mathcal{T}} 
\newcommand{\NT}{\mathcal{NT}} 
\newcommand{\ICP}{\mathcal{ICP}}

\newcommand{\mS}{\mathcal{S}} 

\newcommand{\intv}[2]{\left \{ #1, \dots, #2 \right \}}

\DeclareMathOperator\ply{ply}
\DeclareMathOperator\diam{diam}

\DeclareMathOperator\grid{\boxplus}

\DeclareMathOperator\fvsp{\textsf{fvs}}

\DeclareMathOperator\tw{\textsf{tw}}

\DeclareMathOperator\CPi{\mathcal{P}}

\newcommand{\trh}{\textsc{Triangle Hitting}\xspace}
\newcommand{\wtrh}{\textsc{Weighted Triangle Hitting}\xspace}
\newcommand{\fvs}{\textsc{Feedback Vertex Set}\xspace}
\newcommand{\oct}{\textsc{Odd Cycle Transversal}\xspace}
\newcommand{\FVS}{FVS\xspace}
\newcommand{\PSEUDO}{\textsc{Pseudo Forest Del}\xspace}
\newcommand{\PtH}{P$_t$-\textsc{Hitting}\xspace}
\renewcommand{\TH}{TH\xspace}
\newcommand{\OCT}{OCT\xspace}
\newcommand{\THREESAT}{3-SAT\xspace}
\newcommand{\PTHREESAT}{\textsc{Planar-3-SAT}\xspace}

\newcommand{\ruleref}[1]{(\hyperref[#1]{R\ref*{#1}})}

\renewcommand{\O}{\mathcal{O}}
\renewcommand{\o}{o}
\newcommand{\Ostar}{\O^*}

\DeclareMathOperator\Nstar{N^{\star}}
\DeclareMathOperator\NSstar{N_\mS^{\star}}
\DeclareMathOperator{\ddstar}{\mu^{\Nstar}\!}
\DeclareMathOperator\lr{lr}

\newcommand{\NP}{NP\xspace}
\newcommand{\ETH}{ETH\xspace}
\newcommand{\SQGM}{SQGM\xspace}
\newcommand{\ASQGM}{ASQGM\xspace}

\newcommand{\dDIR}{$d$-DIR\xspace}
\newcommand{\DEUXDIR}{2-DIR\xspace}
\newcommand{\CONTACTSEG}{contact-segment\xspace}
\newcommand{\CONTACTDEUXDIR}{contact-\DEUXDIR{}}

\def\cqedsymbol{\ifmmode$\lrcorner$\else{\unskip\nobreak\hfil
\penalty50\hskip1em\null\nobreak\hfil$\lrcorner$
\parfillskip=0pt\finalhyphendemerits=0\endgraf}\fi} 
\newcommand{\cqed}{\renewcommand{\qed}{\cqedsymbol}}

\graphicspath{{fig/}}

\title{Subexponential algorithms in geometric graphs via the subquadratic grid minor property: the role of local radius}

\titlerunning{Subexp. algo. in geometric graphs via the SQGM property: the role of local radius}

\author{Gaétan Berthe}{LIRMM, Université de Montpellier, CNRS, Montpellier, France.}{}{https://orcid.org/0000-0003-0017-6922}{}
\author{Marin Bougeret}{LIRMM, Université de Montpellier, CNRS, Montpellier, France.}{}{https://orcid.org/0000-0002-9910-4656}{}
\author{Daniel Gonçalves}{LIRMM, Université de Montpellier, CNRS, Montpellier, France.}{}{https://orcid.org/0000-0003-3228-9622}{}
\author{Jean-Florent Raymond}{CNRS, LIP, Lyon, France.}{}{https://orcid.org/0000-0003-4646-7602}{Supported by the ANR project GRALMECO (ANR-21-CE48-0004).}

\authorrunning{G. Berthe, M. Bougeret, D. Gonçalves, J-.F. Raymond}

\Copyright{Gaétan Berthe, Marin Bougeret, Daniel Gonçalves and Jean-Florent Raymond}

\ccsdesc[500]{Mathematics of computing~Graph algorithms}
\ccsdesc[500]{Theory of computation~Fixed parameter tractability}
\ccsdesc[500]{Theory of computation~Computational geometry} 

\keywords{geometric intersection graphs, subexponential FPT algorithms, cycle-hitting problems, bidimensionality} 

\EventEditors{John Q. Open and Joan R. Access}
\EventNoEds{2}
\EventLongTitle{42nd Conference on Very Important Topics (CVIT 2016)}
\EventShortTitle{CVIT 2016}
\EventAcronym{CVIT}
\EventYear{2016}
\EventDate{December 24--27, 2016}
\EventLocation{Little Whinging, United Kingdom}
\EventLogo{}
\SeriesVolume{42}
\ArticleNo{23}

\begin{document}

\maketitle
\begin{abstract}
We investigate the existence in geometric graph classes of subexponential parameterized algorithms for cycle-hitting problems like \textsc{Triangle Hitting} (TH), \textsc{Feedback Vertex Set} (FVS) or \textsc{Odd Cycle Transversal} (OCT). These problems respectively ask for the existence in a graph $G$ of a set $X$ of at most $k$ vertices such that $G-X$ is triangle-free, acyclic, or bipartite. It is know that subexponential FPT algorithms of the form $2^{o(k)}n^{\O(1)}$ exist in planar and even $H$-minor free graphs from bidimensionality theory [Demaine et al. 2005], and there is a recent line of work lifting these results to geometric graph classes consisting of intersection of similarly sized ``fat'' objects ([Fomin et al. 2012]\nocite{Fomin12}, [Grigoriev et al. 2014], or disk graphs [Lokshtanov et al. 2022], [An et al. 2023]). 

In this paper we first identify sufficient conditions, for any graph class $\C$ included in string graphs, to admit subexponential FPT algorithms for any problem in $\CPi$, a family of bidimensional problems where one has to find a set of size at most $k$ hitting a fixed family of graphs, containing in particular \FVS. Informally, these conditions boil down to the fact that for any $G \in \C$, the \emph{local radius of $G$} (a new parameter introduced in [Lokshtanov et al. 2023]) is polynomial in the clique number of $G$ and in the maximum matching in the neighborhood of a vertex.
To demonstrate the applicability of this generic result, we bound the local radius for two special classes: 
intersection graphs of axis-parallel squares and of contact graphs of segments in the plane. This implies that any problem $\Pi \in \CPi$ (in particular, \FVS) can be solved in:
\begin{itemize}
  \item $2^{\O(k^{3/4}\log k)}n^{\O(1)}$-time in contact segment graphs,
    \item $2^{\O(k^{9/10}\log k)}n^{\O(1)}$ in intersection graphs of axis-parallel squares
\end{itemize}
On the positive side, we also provide positive results for \TH by solving it in: 
\begin{itemize}
    \item $2^{\O(k^{3/4}\log k)}n^{\O(1)}$-time in contact segment graphs,
    \item $2^{\O(\sqrt d t^2 (\log t)  k^{2/3}\log k)} n^{\O(1)}$-time in $K_{t,t}$-free \dDIR graphs (intersection of segments with at most $d$ possible slopes)
  \end{itemize}

On the negative side, assuming the \ETH we rule out the existence of algorithms solving:
\begin{itemize}
    \item \TH and \OCT in time $2^{o(n)}$ in \DEUXDIR{} graphs and more generally in time $2^{o(\sqrt{\Delta n})}$ in \DEUXDIR{} graphs with maximum degree $\Delta$, and
    \item \TH, \FVS, and \OCT in time $2^{o(\sqrt{n})}$ in $K_{2,2}$-free contact-\DEUXDIR{} graphs of maximum degree~6.
\end{itemize}
Observe that together, these results show that the absence of large $K_{t,t}$ is a necessary and sufficient condition for the existence of subexponential FPT algorithms for \TH in \DEUXDIR{}.
\end{abstract}

\section{Introduction}
\label{sec:intro}

In this paper we consider fundamental \NP-hard cycle-hitting problems like \trh{} (\TH), \fvs{} (\FVS), and \oct{} (\OCT) where, given a graph $G$ and an integer $k$, the goal is to decide whether $G$ has a set of at most $k$ vertices hitting all its triangles (resp. cycles for \FVS, and odd cycles for \OCT).
We consider these problems from the perspective of parameterized complexity,
where the objective is to answer in time $f(k)n^{\O(1)}$ for some computable function $f$, and with $n$ denoting the order of~$G$. It is known (see for instance~\cite{Cygan2015Book}) that these three problems can be solved on general graphs in time $c^{\O(k)}n^{\O(1)}$ (for some constant $c$) and that, under the Exponential Time Hypothesis (\ETH), the contribution of $k$ cannot be improved to a subexponential function (i.e., there are no algorithms with running times of the form $c^{o(k)}n^{\O(1)}$ for these problems).
However, it was discovered that some problems admit subexponential time algorithms in certain classes of graphs, and there is now  a well established set of techniques to design such algorithms.
Let us now review these techniques and explain why they do not apply on the problems we consider here.

\subparagraph*{Subexponential FPT algorithms in sparse graphs.}
Let us start with the bidimensionality theory, which gives an explanation on the so-called \emph{square root phenomenon} arising for planar and $H$-minor free graphs~\cite{demaine2005subexponential} for bidimensional\footnote{Informally: yes-instances are minor-closed and a solution on the $(r,r)$-grid has size $\Omega(r^2)$.} problems, where a lot of graph problems admit \ETH-tight $2^{\O(\sqrt{k})}n^{\O(1)}$ algorithms.
What we call a \emph{graph parameter} here is a function $p$ mapping any (simple) graph to a natural number and that is invariant under isomorphism.
The classical win-win strategy to decide if $p(G)\le k$ for a minor-bidimensional\footnote{See definition in~\cite{grigoriev2014bidimensionality}.} parameter (like $p=\fvsp$, the size of a minimum feedback vertex set of $G$) is to first reduce to the case where $\grid(G)=\O(\sqrt{k})$ (where $\grid(G)$ denotes the maximum $k$ such that the $(k,k)$-grid is contained as a minor in $G$), 
and then use an inequality of the form $\tw(G) \le f(\grid(G))$ to bound the treewidth obtained through the following property.

\begin{definition}[\cite{baste2022contraction}]
Given $c < 2$, a graph class $\G$ has the \emph{subquadratic grid minor property} for $c$ (\emph{SQGM} for short), denoted $\G \in SQGM(c)$, if $\tw(G)=\O(\grid(G)^c)$ for all $G\in \G$.
We write $\G \in SQGM$ if there exists $c<2$ such that $\G \in \SQGM(c)$.
\end{definition}
While in general every graph $G$ satisfies the inequality $\tw(G) \leq \grid(G)^c$ for some $c<10$~\cite{chuzhoy2021towards}, the SQGM property additionally require that $c<2$.
Thus, for any $\G \in SQGM(c)$ and $G\in \G$ such that $\grid(G)=\O(\sqrt{k})$, we get $\tw(G) \le \grid(G)^c = \O\left (k^{c/2}\right )=o(k)$.
For instance planar graphs and more generally $H$-minor free graph~\cite{demaine2008linearity} are known to have a treewidth linearly bounded from above by the size of their largest grid minor. In other words, these classes belong to $SQGM(1)$.
The conclusion is that the SQGM property allows subexponential parameterized algorithms for minor-bidimensional problems (if the considered problem has a $2^{\O(\tw(G))}n^{\O(1)}$-time algorithm)
on sparse graph classes.
Notice that these techniques have been extended to contraction-bidimensional problems~\cite{baste2022contraction}.

\subparagraph*{Extension to geometric graphs.}
\label{subsec:extending}
Consider now a geometric graph class $\G$, meaning that any $G\in \G$ represents the interactions of some specified geometric objects. We consider here (Unit) Disk Graphs which correspond to intersection of (unit) disks in the plane, \dDIR{} graphs (where the vertices correspond to segments with $d$ possible slopes in $\mathbb{R}^2$), and \CONTACTSEG{} graphs (where each vertex corresponds to a segment in $\mathbb{R}^2$, and any intersection point between two segments must be an endpoint of one of them). We refer to \autoref{sec:graphCl} for formal definitions. 
Classes of geometric graphs represented in the plane form an appealing source of candidates to obtain subexponential parameterized algorithms as there is an underlying planarity in the representation. However these graphs are no longer sparse as they may contain large cliques, and thus cannot have the \SQGM property. Indeed, if $G$ is a clique of size $a$, then $\tw(G)= a-1$ but $\grid(G) \le \sqrt{|G|} = \sqrt{a}$.
To overcome this, let us introduce the following notion where the bound on treewidth is allowed to depend on an additional parameter besides $\grid(G)$.

\begin{definition}
Given a graph parameter $p$ and a real $c < 2$, a graph class $\G$ has the \emph{almost subquadratic grid minor property} (\emph{ASQGM} for short) for $p$ and $c$ if there exists a function $f$ such that $\tw(G)=\O(f(p(G))\grid(G)^c)$.
The class $\G$ has $\ASQGM(p)$ if there exists $c<2$ such that $\G$ has the \ASQGM property for $p$ and $c$. The notation is naturally extended to more than one parameter.
\end{definition}

This notion was used implicitly in earlier work (e.g.,\ \cite{grigoriev2014bidimensionality}) but we chose to define it explicitly in order to highlight the contribution $f$ of the parameter $p$ to the treewidth, which is particularly relevant when it can be shown to be small (typically, polynomial).
Let us now explain how \ASQGM can be used to obtain subexponential parameterized algorithms on geometric graphs.

It was shown in \cite{fomin2018excluded} that \FVS can be solved in time $2^{\O(k^{3/4}\log k)}n^{\O(1)}$ in map graphs, a superclass of planar graphs where arbitrary large cliques may exist, as follows. Let $\omega(G)$ denote the order of the largest clique in a graph $G$.
The first ingredient is to prove that map graphs have $\ASQGM(\omega)$, and more precisely that $\tw(G) = \O(\omega(G)\grid(G))$.
Then, if $\omega(G)\ge k^\epsilon$ for some $\epsilon$, the presence of such large clique allows to have subexponential branchings (as a solution of \FVS must take almost all vertices of a clique).
When $\omega(G) < k^\epsilon$, then the \ASQGM property gives that $\tw(G) \le k^{\epsilon} \grid(G) \le k^{\frac{1}{2}+\epsilon}$ (as before we can immediately answer no if $\grid(G) > \O(\sqrt{k})$). 
By appropriately choosing $\epsilon$ the authors of \cite{fomin2018excluded} obtain the mentioned running time.
The same approach also applies to unit disk graphs and has since been improved to $2^{\sqrt{k}\log k}n^{\O(1)}$ in~\cite{fomin2019finding} using a different technique, and finally improved to an optimal $2^{\sqrt{k}}(n+m)$ in~\cite{an2021feedback} for similarly sized fat objets (which typically includes unit squares, but not disks, squares, nor segments).

There is also a line of work aiming at establishing \ASQGM property for different classes of graphs and parameters, with for example~\cite{grigoriev2014bidimensionality} proving that (1) string graphs have \ASQGM when the parameter $p$ is the number of times a string is intersected (assuming at most two strings intersect at the same point),  and that (2) intersection graphs of ``fat'' and convex objects have \ASQGM when the parameter $p(G)$ is the minimal order of a graph $H$ not subgraph of $G$ (generalizing the degree when $H$ is a star).

\subparagraph*{When \texorpdfstring{$ASQGM(\omega)$}{ASQGM(ω)} does not hold.}

A natural next step for \FVS and \TH is to consider classes that are not $ASQGM(\omega)$. Observe (see \autoref{fig:not_asqgm_intro}) that neither disk graphs, nor \CONTACTDEUXDIR{} graphs are in $ASQGM(\omega)$, and thus constitute natural candidates.
\begin{figure}[!ht]
    \centering
    \includegraphics[width=0.8\textwidth]{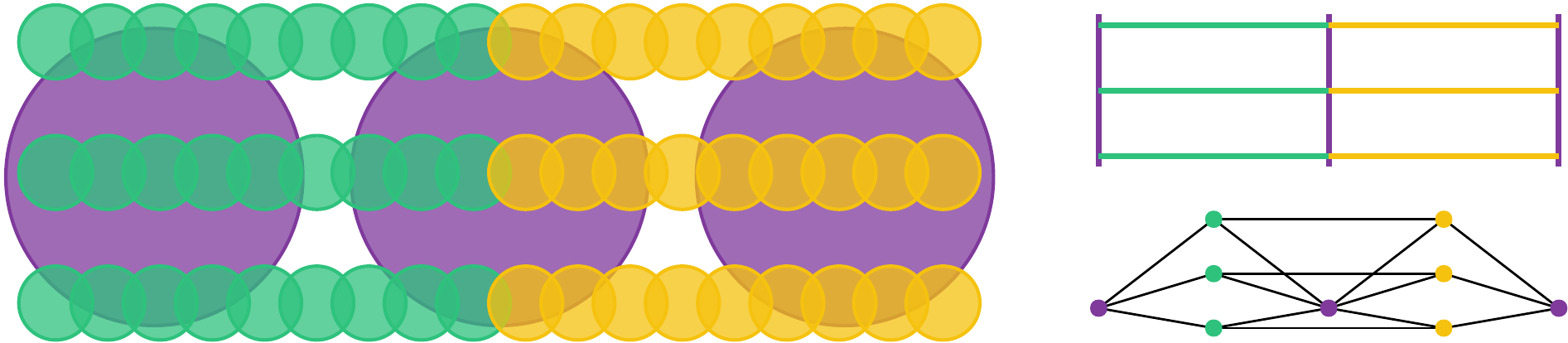}
    \caption{Left: a representation of a disk graph. Right: a contact \DEUXDIR{} graph and the corresponding graph. In these graphs (where the left one is from~\cite{fomin2018excluded}), $\omega(G)$ is constant, $\tw(G) \ge t$ (where $t=3$ here) as it contains $K_{t,t}$ as a minor, and $\grid(G) = \O(\sqrt{t})$  as they have a feedback vertex set of size at most $t$.}
    \label{fig:not_asqgm_intro}
\end{figure}

New ideas allowed the authors of~\cite{lokSODA22} to obtain subexponential parameterized algorithms on disk graphs, in particular for \TH and \FVS.
The first idea is a preliminary branching step (working on general graphs) which given an input $(G,k)$ first reduces to the case where we are given a set $M$ of size $\O(k^{1+\epsilon})$ such that $G - M$ is a forest and, for any $v \in M$, $N(v) \setminus M$ is an independent set
(corresponding to \autoref{cor:bothbranchings}, but where we consider a generic problem instead of \FVS).
The second idea is related to neighborhood complexity which, informally, measures the number of ways the vertices of $G-X$ connect to the vertices of $X$ for every vertex set $X$ (see \autoref{def:neighc} for a formal definition).
The following theorem was originally formulated using ply (the maximum number of disks containing a fixed point) instead of clique number, but it is known~\cite{bonamy2018eptas} that these two values are linearly related in disk graphs.

\begin{theorem}[Theorem 1.1 in \cite{lokSODA22}]
\label{thm:DGncomplexity}
Disk graphs with bounded clique number have linear neighborhood complexity.
\end{theorem} 
For \TH, these two ideas are sufficient to obtain a subexponential parameterized algorithm.
For \FVS, \cite{lokSODA22} provides the following corollary.


\begin{corollary}[Corollary 1.1 in~\cite{lokSODA22} restricted to \FVS]
\label{cor:DGsmalltw}
Let $G$ be a disk graph with a (non-necessarily minimal) feedback vertex set $M \subseteq V(G)$ such that for
all $v \in M$, $N(v) \setminus M$ is an independent set, and such that for all  $v\in V(G) \setminus M$, $N(v)\setminus M$ is non-empty. Then, the treewidth of $G$ is $\O(\sqrt{|M|}\omega(G)^{2.5})$.
\end{corollary}

As they use this corollary after a branching process reducing the clique number to $k^{\epsilon}$ and as their (approximated) feedback vertex set $M$ has size $|M|=k^{1+\epsilon'}$, they obtain a sublinear treewidth and thus a subexponential parameterized algorithm for \FVS (and several variants of \FVS) running in time $2^{\O(k^{13/14}\log k)}n^{\O(1)}$.
Recently this running time has been improved to $2^{\O(k^{7/8}\log k)}n^{\O(1)}$ when the representation is given and $2^{\O(k^{9/10}\log k)}n^{\O(1)}$ otherwise~\cite{Faster2023Shinwoo}.

\subparagraph*{Subexponential FPT algorithms via kernels.}

Another approach to obtain $2^{o(k)}n^{\O(1)}$ algorithms is to obtain small kernels (meaning computing in polynomial time an equivalent instance $(G',k')$
with $|G'|$ typically in $\O(k)$), and then use a $2^{o(n)}$ time algorithm. For \FVS such a $2^{o(n)}$-time algorithm is known in string graphs from~\cite{bonnet2019optimality} or \cite{novotna2021subexponential}, and was recently generalized to induced-minor-free graph classes~\cite{korhonen2023induced}.
However, as far as we are aware, the existence of a subquadratic kernel in this graph class is currently open.

\subsection{Our contribution}

Our objective is to study the existence of subexponential parameterized algorithms for hitting problems like \FVS and \TH in different types of intersection graphs. Our algorithmic results are summarized in \autoref{fig:results}.

\begin{table}
\begin{center}
\renewcommand{\arraystretch}{0.8}

\begin{tabular}{|c|c|c|c|c|}
\hline

    \multicolumn{5}{|c|}{Upper bounds} \\
    \hline
     Restriction & of class & Problem & Time complexity &  Section\\
     \hline

     none& square graphs & \multirow{2}*{$\Pi \in \CPi$} & $2^{\O(k^{9/10})}n^{\O(1)} $&  \autoref{sec:asqgm}\\
     \cline{1-2}\cline{4-5}

     \multirow{2}*{contact} &\multirow{2}*{segment graphs}&   & $2^{\O(k^{7/8} \log k)} n^{\O(1)}$ & \autoref{sec:asqgm-contact-seg}\\
    \cline{3-5}
      &   &\multirow{3}*{\TH}& \multirow{1}*{$2^{\O(k^{3/4} \log k)}n^{\O(1)}$}&\multirow{3}*{\autoref{sec:sublin}} \\
    \cline{1-2}\cline{4-4}
     \multirow{2}*{$K_{t,t}$-free} & \dDIR graphs && $2^{\O(k^{2/3}(\log k)  \sqrt d t^2 \log t )} n^{\O(1)}$&\\
     \cline{2-2}\cline{4-4}
      & string graphs & & $2^{\O_t(k^{2/3} \log k)} n^{\O(1)}$&\\

     \hline
\end{tabular}

\vspace{1em}
\begin{tabular}{|c|c|c|c|c|c|}
\hline
    \multicolumn{5}{|c|}{Lower bounds (under \ETH)} \\
    \hline
     Restriction & of class & Problem & Lower bound & Section\\
     \hline
     none& \multirow{3}*{\DEUXDIR} & \multirow{2}*{\TH, \OCT} & $2^{\o(n)}$ & \multirow{1}*{\autoref{sec:negative}} \\
     \cline{1-1}\cline{4-5}
     Maximum degree $\Delta$, for $\Delta\geq 6$&  &  &$2^{\o\left(\sqrt {\Delta n}\right)}$ &\multirow{2}*{\autoref{sec:negative}}  \\
     \cline{1-1}\cline{3-4}
     $K_{2,2}$-free contact, maximum degree $6$ &  & \TH, \FVS, \OCT & $2^{\o\left(\sqrt n\right)}$ &\\
     \hline
\end{tabular}
\renewcommand{\arraystretch}{1}
\end{center}

\caption{Summary of our results. All algorithms are robust, i.e., they do not need a representation.}
\label{fig:results}
\end{table}

\subparagraph*{Positive results via \ASQGM.}
In \autoref{sec:asqgm} we explain how the local radius (hereafter denoted $\lr$), introduced recently in \cite{Lokshtanov23Approx} in the context of approximation, can be used to get subexponential FPT algorithms for any problem in $\CPi$, a family of bidimensional problems where one has to find a set of size at most $k$ hitting a fixed family of graphs.
This class contains in particular \FVS, and \PSEUDO (resp. \PtH) where given a graph $G$, the goal is to remove a set $S$ of at most $k$ vertices of $G$ such that each connected component of $G-S$ contains at most one cycle (resp. does not contain a path on $t$ vertices as a subgraph).
We point out that these three problems are also in the list of problems mentioned in~\cite{Lokshtanov23Approx} that admit EPTAS in disk graphs.
We first provide sufficient conditions for graph class to admit subexponential FPT algorithms for any problem in $\CPi$, after the preprocessing step of \autoref{cor:bothbranchings} (introduced for disk graphs in~\cite{lokSODA22}) has been performed.
These conditions mainly boil down to having $\ASQGM(\omega,\ddstar)$, where $\ddstar$ is, informally, the maximum size of matching in the neighborhood of a vertex.
Then, we use the framework of \cite{baste2022contraction} to show that string graphs have $\ASQGM(\omega,\lr)$.
Thus, the message of \autoref{sec:asqgm} is that in order to obtain a subexponential FPT algorithm for a problem $\Pi \in \CPi$ in a given subclass of string graphs,  the only challenge is to bound the local radius by a polynomial of $\omega$ and $\ddstar$.
Finally, we provide such bounds for square graphs (intersection of axis-parallel squares) and \CONTACTSEG graphs.

We point out that in our companion paper \cite{gbgr2023fvspseudo} we prove that \FVS admits an algorithm running in time $2^{\O(k^{10/11} \log k)} n^{\O(1)}$ for pseudo-disk graphs. As square and segment graphs are in particular pseudo-disk graphs, this generalizes the graph class where subexponential parameterized algorithms exist, but to the price of a worst running time.
Moreover, our result in~\cite{gbgr2023fvspseudo} is obtained via kernelization techniques which require a representation of the input graph (i.e., this algorithm is not robust), and the reduction rules behind the kernel are tailored for \FVS and not applicable for any problem $\Pi \in \CPi$.

\subparagraph*{Negative results.}
An interesting difference between disk graphs and $d$-DIR graphs is that \autoref{thm:DGncomplexity} (about the linear neighborhood complexity) no longer holds for $d$-DIR graphs, because of the presence of large bicliques. 
Thus, it seems that $K_{t,t}$ is an important subgraph differentiating
the two settings and this fact is confirmed by the two following
results. First we show in \autoref{sec:negative} that assuming the \ETH, there is no algorithm solving \TH and \OCT in time $2^{o(n)}$ on $n$-vertex $2$-DIR graphs and more generally in time $2^{o(\sqrt{\Delta n})}$ in \DEUXDIR{} graphs with maximum degree $\Delta$. 
We note that the result for \OCT was already proved in 
\cite{okrasa2020subexponential} as a consequence of algorithmic lower bounds for homomorphisms problems in string graphs.
In our second negative result, we prove that assuming the \ETH, the problems \TH, \OCT, and \FVS cannot be solved in time $2^{o\left (\sqrt{n}\right )}$ on $n$-vertex $K_{2,2}$-free 
   \CONTACTDEUXDIR{} graphs. 
Notice that that our $2^{o\left (\sqrt{n}\right )}$ lower-bounds match those known for the same problems in planar graphs~\cite{cai2003eth}.

\subparagraph*{Positive results for \TH.}
In \autoref{sec:sublin} we observe that, for any hereditary graph class with sublinear separators, the preliminary branching step in \autoref{cor:bothbranchings} of~\cite{lokSODA22} directly leads to a subexponential parameterized algorithm for \TH. This implies the $2^{c_t k^{2/3} \log k} n^{\O(1)}$ algorithm for $K_{t,t}$-free string graphs.
Recall that according to our negative result in \autoref{sec:negative},  the $K_{t,t}$-free assumption is necessary.
To improve the constant $c_t$ in special cases, we provide in \autoref{ssec:neighcplx}  bounds on the neighborhood complexity of two subclasses that may be of independent interest: $K_{t,t}$-free
\dDIR{} graphs have linear neighborhood complexity with ratio $\O(d t^3 \log t)$, and \CONTACTSEG graphs have linear neighborhood complexity.
These bounds lead to improved running times for \TH in the corresponding graph classes (see \autoref{fig:results}).
\medskip

\section{Preliminaries}
\label{sec:prelim}

\subsection{Basics}
In this paper logarithms are binary and all graphs are simple, loopless and undirected.
Unless otherwise specified we use standard graph theory terminology, as in \cite{diestel2005graph} for instance.
Given a graph $G$, we denote by $\omega(G)$ the maximum order of a clique in $G$.
We denote by $d_G(v)$ the degree of $v \in V(G)$, or simply $d(v)$ when $G$ is clear from the context.
    The \emph{distance} between two vertices of a graph is the minimum length (in number of edges) of a path linking them, and the \emph{diameter} of a graph is the maximum distance between two of its vertices.
The \emph{radius} of a graph is the smallest integer $r \ge 0$ such that there exists a vertex $v$ such that every vertex in the graph is at a distance at most $r$ from $v$. A \emph{$t$-bundle}~\cite{Lokshtanov23Approx} is a matching of size $t$ plus a vertex connected to the $2t$ vertices of the matching. We say that $B$ is a $t$-bundle of a graph $G$ if $G[B]$ is a $t$-bundle plus possibly some extra edges.
A set $S \subseteq V(G)$ is a $t$-bundle hitting set of $G$ if $S \cap B \neq \emptyset$ for any $t$-bundle $B$ of $G$.
We denote by $\grid(G)$ the maximum $k$ such that the $(k,k)$-grid is contained as a minor in~$G$.
We denote by $\tw(G)$ the treewidth of $G$, and $\mu(G)$ the size of a maximum matching of $G$.

In \autoref{sec:asqgm} we provide subexponential parameterized algorithms for a class of problems $\CPi$ that we will now define.
We restrict our attention to \emph{hitting problems}, where for a fixed graph family $\mathcal{F}$, the input is a graph $G$ and an integer $k$, and the goal is to decide if there exists $S \subseteq V(G)$ with $|S| \le k$ such that $G-S \in \mathcal{F}$. A general setting where our results hold is described by the class $\CPi$ defined below and inspired by the problems tackled in \cite{Lokshtanov23Approx}
.

\begin{definition}
We denote by $\CPi$ the class of all hitting problems $\Pi$ such that:
\begin{enumerate}
    \item $\Pi$ is bidimensional ;
    \item there is an integer $c_\Pi > 0$ such that for any solution $S$ in a graph $G$, and any $c_\Pi$-bundle $B$ of $G$, $S \cap B \neq \emptyset$; and
    \item $\Pi$ can be solved on a graph $G$ in time $\tw(G)^{\O(\tw(G))}$.
    \end{enumerate}
\end{definition}

\begin{claim}\label{claim:pbincpi}
\FVS, \PSEUDO and \PtH for $t \le 5$ belong to $\CPi$.        
\end{claim}
\begin{proof}
It is well known that these three problems are bidimensional.
For the second condition, one can check that $c_\Pi$ is equal to $1$ for \FVS (as a $1$-bundle is a triangle) and equal to $2$ for \PSEUDO and \PtH when $t \le 5$.
For the last condition, as \FVS corresponds to hit all $K_3$ as minor and \PSEUDO correspond to hit all $\{H_0,H_1,H_2\}$ as a minor (with $H_i$ is formed by two triangles sharing $i$ vertices), these two problems can be solved in $\tw(G)^{\O(\tw(G))}$ by~\cite{baste2019hitting}.
For \PtH the result holds by~\cite{CYGAN201762}.
\end{proof}

\subsection{Graph classes}
\label{sec:graphCl}

A summary of graph classes considered in this article is presented in \autoref{fig:graphclasses}.

\begin{figure}[!ht]
\centering           
    \includegraphics[scale=1]{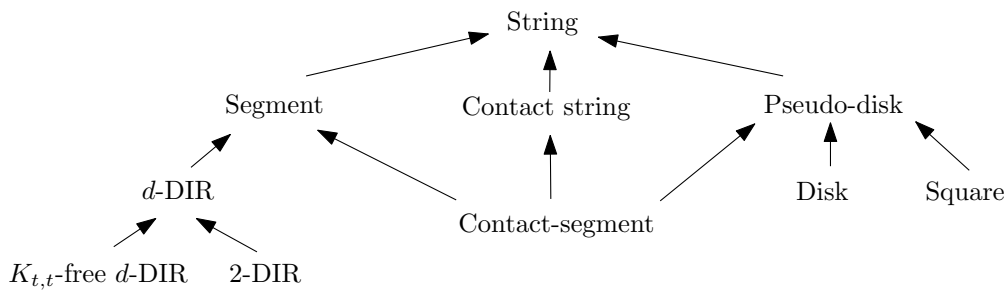}
    \caption{Inclusion between graph classes.}
    \label{fig:graphclasses}
\end{figure}

 \begin{figure}[!ht]
\centering           
    \includegraphics[scale=0.6]{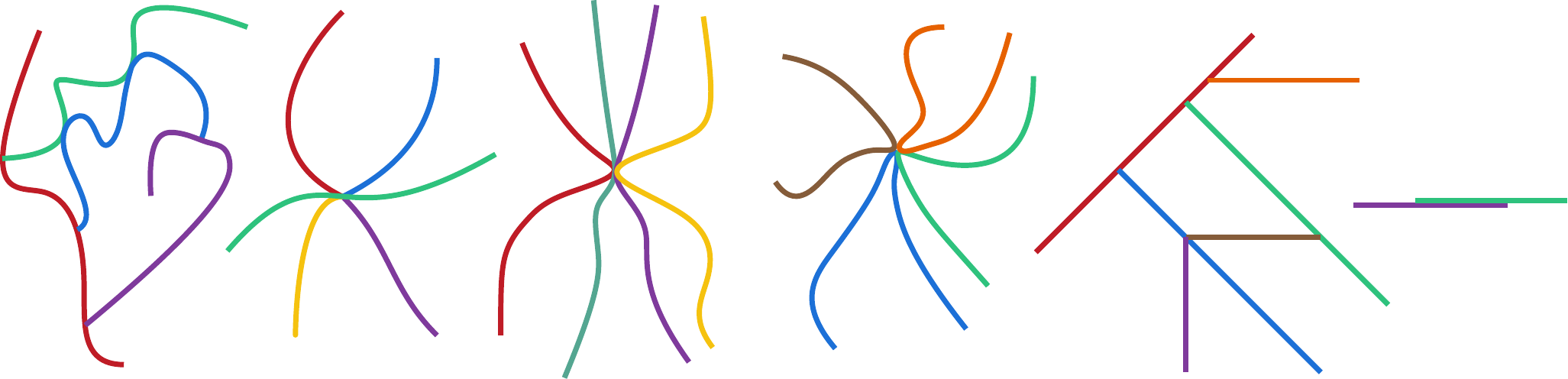}
    \caption{From left to right, four representations of contact string graphs, then a representation of 3-DIR \CONTACTSEG graph, and finally on the right an example of an intersection between segments not allowed in a representation of a \CONTACTSEG graph. }
    \label{fig:examples}
\end{figure}

 In this article,  we are mainly concerned with geometric graphs described by the intersection or contact of objects in the Euclidean plane.
 The most general class we consider are \emph{string graphs}, which are intersection graphs of strings (a.k.a. Jordan arcs).
Intersection graphs of segments in $\RR^2$ are called \emph{segment graphs}.
If a segment graph can be represented with at most $d$ different slopes, we call it a \emph{\dDIR{} graph}.\footnote{In general two \dDIR{} graphs may require different sets of slopes in their representation but in the case $d=2$ it is known that the segments can be assumed to be axis-parallel, which we will do.}
These classes of intersection graphs admit \emph{contact} subclasses, where the representations should not contain \emph{crossings}. That is, two strings either intersect tangentially, or they intersect at an endpoint of one of them. 
In a segment contact representation, any point belonging to two segments must be an endpoint of at least one of these segments.  
If a point belongs to several strings or segments, the above property must hold for any pair of them. 
This defines \emph{contact string graphs}, \emph{\CONTACTSEG graphs} and \emph{contact \dDIR graphs}.

\section{Preliminary branching steps}
\label{sec:branching}

Our algorithms make use of two preprocessing branchings: the first one is a folkore branching that allows to reduce cliques larger than a chosen size $p$ (where typically $p=k^{\epsilon}$), and the second allows us to branch on large bundles. These steps are described in \cite{lokSODA22} for \FVS in disk graphs. Here we extend them to any problem in $\CPi$ and any graph class where the maximum clique can be approximated in polynomial time. Their proofs, which we include for completeness, can be obtained by closely following the proofs in~\cite{lokSODA22}.

\begin{remark}
  In the rest of the section we refer to the full version \cite{lok2021complete} (availlable online) of the extended abstract \cite{lokSODA22}.
\end{remark}

\begin{lemma}[Adapted from {\cite[Lemma 6.2]{lok2021complete}}]\label{lem:branchingClique}
Let $\Pi \in \CPi$ and $\G$ be a hereditary graph class where the maximum clique can be $\alpha$-approximated within a constant factor $\alpha\geq 1$ in polynomial time.
There exists a $2^{\O\left (\frac{k}{p}\log p\right )}n^{\O(1)}$-time algorithm that, given a 
graph $G \in \G$ and $p,k \in \mathbb{N}$
with $p \ge 6\alpha c_{{\Pi}}$, returns a collection $\mathcal{Y} \subseteq \{(D, U) : D,U \subseteq V(G), D \cap U = \emptyset \}$ of size $2^{\O\left (\frac{k}{p}\log p \right)}$ such that:
\begin{enumerate}
\item For every $(D,U) \in \mathcal{Y}$, $G - D$ does not have a clique of size larger than $p$.
\item For every solution $S \subseteq V(G)$ of size at most $k$, there exists a unique
$(D,U) \in \mathcal{Y}$ such that $D \subseteq S$ and $S \cap U = \emptyset$.
\end{enumerate}
\end{lemma}
\begin{proof}
    The algorithm is the same than the one in \cite[Lemma 6.2]{lok2021complete} with minor modifications in the branching steps. In this proof we use the same notation. We only consider the differences in the analysis of the branching. Observe that for every clique $C$ of $G$, a solution $S$ of the problem $\Pi$ will contain at least $|C|-2c_\Pi$ vertices of $C$ as otherwise a $c_\Pi$-bundle would not be hit.
    For some $x>p/\alpha$, the algorithm makes at most $\sum_{i=0}^{2c_{\Pi}}\binom{x}{i}\leq x^{2c_\Pi}$ recursive calls, each with the parameter $k'$ decreased by at least $x-2c_{\Pi}$. We denote $T(k')$ the time complexity of a call to the algorithm with parameter $k'$. If $k'\leq p/\alpha - 2c_{\Pi}$, a recursive call would make the new parameter strictly below $0$, so there is no such call and $T(k')=n^{\O(1)}$. Otherwise for $k\geq p/\alpha - 2c_{\Pi}$ we have 
    \begin{align*}
        T(k')&=\max_{x>p/\alpha\text{ s.t. }k'-(x-2c_{\Pi})\geq 0}x^{2c_{\Pi}}\cdot T(k'-(x-2c_{\Pi}))+n^{\O(1)} \\
        &= \left(\frac p \alpha\right)^{2c_\Pi}\cdot T\left (k'-\left (\frac p\alpha-2c_{\Pi}\right )\right )+n^{\O(1)}.
    \end{align*}

    As $p\geq 6\alpha c_\Pi$, we have $\frac p\alpha-2c_{\Pi}\geq \frac {2p}{3\alpha}$ and so $T(k')=\left(\frac p \alpha\right)^{2c_\Pi}\cdot T(k'-\frac {2p}{3\alpha})+n^{\O(1)}.$

    This solves to $T(k)=\left(\left(\frac p \alpha\right)^{2c_\Pi}\right)^{\O\left (\frac{3\alpha k}{2p}\right )}\cdot n^{\O(1)}=2^{\O\left (\frac kp \cdot \log p \right )}\cdot n^{\O(1)}$.

    The same arguments give the bound on the size of the obtained collection.
\end{proof}

\begin{lemma}[Adapted from {\cite[Lemma 6.10]{lok2021complete}}]\label{lem:branching}
Let $\Pi \in \CPi$.
    There is a $2^{ \O\left (\frac{k}{p} \log k\right)} n^{\O(1)}$-time algorithm that, given a $n$-vertex graph $G$, $p,k \in \NN$ with $p\geq 2c_{\Pi}$, and a $c_\Pi$-bundle hitting set $X\subseteq V(G)$, returns a collection
    \[
        \mathcal{Z} \subseteq \{(D,U,Z) : D,U\subseteq X,\ D\cap U = \emptyset,\ \text{and}\ Z \subseteq V(G)\setminus X\}
    \]
    of size $2^{O\left (\frac{k}{p} \log k\right )}$ and such that:
    \begin{enumerate}
        \item \label{e:sizez} for every $(D,U,Z) \in \mathcal{Z}$, $|Z| \leq 4(k + p |X\setminus D|)$ and for every $v \in  X \setminus D$, $\mu(G[N(v) \setminus (X\cup Z)]) < c_{\Pi}$;
        \item for every solution $S \subseteq V(G)$ of size at most $k$, there exists a unique $(D,U,Z) \subseteq \mathcal{Z}$ such that $D\subseteq S$ and $S\cap U = \emptyset$.
    \end{enumerate}
\end{lemma}

\begin{proof}

The algorithm is the same than the one in \cite[Lemma 6.10]{lok2021complete} with minor modifications in the branching steps, and we use the same notations. We only consider the differences in the analysis of the branching. Observe that for a $p$-bundle around a vertex $v$, a solution $S$ of the problem $\Pi$ will either contains $v$ or at least $p-c_\Pi+1$ vertices of the bundle as otherwise a $c_\Pi$-bundle would not be hit.

We denote $T(k')$ the time complexity of a call to the algorithm with parameter $k'$. If $k'=0$ we have $T(k')=n^{O(1)}$ and for $k'\leq p-c_{\Pi}$, a recursive call from a $x$-bundle with $x>p$ corresponding to the branch where the central vertex is not part of the solution would make the parameter strictly below $0$. So there is no such call, and so for $1 \leq k' < p-c_{\Pi}+1$, $T(k')=T(k'-1)+n^{\O(1)}$. For $k'>p-c_{\Pi}$, 
\begin{align*}
    T(k')&=T(k'-1)+T(k'-(p-c_\Pi+1))+n^{\O(1)}\\
    &=T(k'-1)+T(k'-\frac p2)+n^{\O(1)}.
\end{align*}
This resolves to  $T(k')=2^{\O(\frac kp)}n^{\O(1)}$. Using the same arguments we get the bound on the size of the obtained collection.

Another difference with the original algorithm is the bound on the number of vertices in the set $Z$ obtained with this construction. When a set $Z'$ is updated by adding $y\geq 2p$ vertices, the parameter $k'$ decreases by at least $\frac y2-c_\Pi+1$, and having $y\geq 4c_\Pi$ this is at least $\frac y4$. So for every obtained $Z$ we have $|Z|\leq 4k$.
\end{proof}

These two reduction steps are summarized in the following corollary, which is the routine that we will use in our algorithms.

\begin{corollary}\label{cor:bothbranchings}
Let $\Pi \in \CPi$.
Let $\G$ be a hereditary graph class where the maximum clique can be $\alpha$-approximated for some constant factor $\alpha\geq 1$ in polynomial time. 
There exists a $2^{\O\left (\frac{k}{p}\log k \right )}n^{\O(1)}$-time algorithm that, given an instance $(G,k)$ of $\Pi$ and an integer $p\in [6\alpha c_{\Pi},k]$, where $G \in \G$, returns a collection
$\C$ of size $2^{\O\left (\frac{k}{p}\log k \right)}$ of tuples $(G',M,k')$ such that:
\begin{enumerate}
    \item For any $(G',M,k') \in \C$, $(G',k')$ is an instance of $\Pi$ where $G'$ is an induced subgraph of $G$, $\omega(G') \le p$, and $k' \le k$;
    \item $M$ is a $c_\Pi$-bundle hitting set of $G'$ with $|M| = \O(p k)$, and for any $v \in M$, $\mu(G'[N(v)\setminus M]) < c_\Pi$; and
    \item $(G,k)$ is a yes-instance of $\Pi$ if and only if there exists $(G',M,k') \in \C$ such that $(G',k')$ is a yes-instance of $\Pi$.
\end{enumerate}
\end{corollary}

\begin{proof}
We run the algorithm of \autoref{lem:branchingClique}. For any $(D',U') \in \mathcal{Y}$ we do the following. Let $G'=G-D'$, $k'=k-|D'|$.
Observe that for any constant $t$, the problem of hitting all $t$-bundles has a trivial $(2t+1)$-approximation by finding a maximal vertex disjoint packing of $t$-bundles and taking all vertices of this packing. Thus, we compute in polynomial time a $\O(1)$-approximation $X$ of a minimum $c_\Pi$-bundle hitting of $G'$. As the optimal value for $\Pi$ is larger than the minimum $c_\Pi$-bundle hitting set, if $|X|>(2c_\Pi+1) k$ then we can define $\C$ as a singleton containing a dummy no-instance, and thus we now assume that $|X| \le (2c_\Pi+1) k=\O(k)$. Now, we use \autoref{lem:branching} on $(G', k')$ with $p$ and $X$ to obtain a collection $\mathcal{Z}$ as in its statement.
 For any $(D,U,Z) \in \mathcal{Z}$,
define $G''=G'-D$, $k''=k'-|D|$, $M=(X\setminus D) \cup Z$. Call $\C$ the set of the triples $(G'', M, k'')$ obtained that way. From \autoref{lem:branchingClique} and \autoref{lem:branching}, $|\C| \leq 2^{\O\left (\frac{k}{p} \log k \right )}$ and it can be constructed in $2^{\O\left (\frac{k}{p} \log k \right )} n^{\O(1)}$-time.

First and third property are straightforward. Let us consider the second property.
As $X$ is a $c_{\Pi}$-bundle hitting set of $G'$, $X \setminus D$ is 
a $c_{\Pi}$-bundle hitting set of $G''$, and thus $M$ is a $c_{\Pi}$-bundle hitting set of $G''$, and  $|M| \le |X|+|Z| \le |X|+4(k+p|X|)=\O(kp)$.
It remains to prove that for any $v \in M$, $\mu(G[N(v)\setminus M]) < c_\Pi$.
If $v \in Z$, then $v \notin X$, and thus if $N(v)\setminus M \ge c_\Pi$ then 
$v \cup (N(v)\setminus M)$ would be a $c_\Pi$-bundle in $G'$ not intersected by $X$, a contradiction. Otherwise ($v \notin Z$), then $v \in X \setminus D$, and
Item~\ref{e:sizez} of \autoref{lem:branching} implies the desired inequality.
\end{proof}

\section{Positive results via \ASQGM}
\label{sec:asqgm}



\subsection{From \texorpdfstring{$\ASQGM(\omega,\ddstar)$}{ASQGM(ω,μ*)} to subexponential algorithms.}\label{subsec:asqgm1}

In this section we provide subexponential paramterized algorithms for problems of $\CPi$ in any class that has the $\ASQGM(\omega,\ddstar)$ property.

\begin{definition}
    Given a graph $G$, a \emph{subneighborhood function} of $G$ is any function $\Nstar:V(G)\rightarrow 2^{V(G)}$ such that for any $v\in V(G)$, $\Nstar(v) \subseteq N(v)$.    
    Moreover, if for any $u\in V(G)$, $|\{v\in V(G),~ u\in \Nstar(v)\}| \le c$ for some $c\in \mathbb{N}$ then we say that $\Nstar$ has \emph{$c$-bounded occurrences}.
    
    Given a subneighborhood function $\Nstar$, we define $\ddstar (v)$ as the maximum number of edges of a matching in $G[\Nstar(v)]$. We denote by $\ddstar(G)$ the maximum of $\ddstar$ over $V(G)$.
\end{definition}
For example in square graphs, we will fix a representation $\mS$, and define $\Nstar(v)$ as the set of neighbors of $v$ whose square is smaller than the one of $v$.

The main theorem from this subsection is the following. Recall that 
$\CPi$ encompasses fundamental algorithmic problems such as \FVS, \PSEUDO and \PtH for $t \le 5$ (\autoref{claim:pbincpi}).

\begin{restatable}{theorem}{subexpASQGM}
    \label{th:subexpASQGM}
    Let $\Pi$ be a problem of $\CPi$ and $\mathcal{C}$ be a hereditary graph class such that:
    \begin{itemize}
    \item maximum clique can be $\O(1)$-approximated in polynomial time in $\mathcal{C}$;
    \item for any $G \in \mathcal C$, there exists a subneighborhood function $\Nstar$ that has $\O(\omega(G)^{c_1})$-bounded occurrences for some $c_1 \in \mathbb{N}$; and
    \item $\C$ has the $\ASQGM(\omega,\ddstar)$ property, i.e., there exists a multivariate polynomial $P$  such that for all $G \in \mathcal C$, we have $\tw(G)=\O(P(\omega(G), \ddstar(G)) \cdot \grid(G))$.
    \end{itemize}
Then,  $\Pi$ admits a parameterized subexponential algorithm on $\C$.
More precisely, for $\epsilon>0$ such that $P(k^\epsilon, k^{(c_1+2)\epsilon})=\O(k^{\frac 12-\epsilon})$, $\Pi$ admits a parameterized subexponential algorithm on $\C$ running in time $2^{\O(k^{1-\epsilon}\log(k))}$.
 This algorithm does not need a representation except if one is required for finding the $\O(1)$-approximation of a maximum clique.
\end{restatable}


\begin{lemma}
\label{lm:removeBIG}
Let $\Pi$ be a problem of $\CPi$.
Consider a graph $G$ and $\Nstar$ a $c$-bounded occurrences subneighborhood function of $G$.
Let $M\subseteq V(G)$ be a $c_\Pi$-bundle hitting set of $G$ such that for any vertex $v\in M$, $\mu(G[N(v)]-M) < c_\Pi$. Then for every positive integer $\tau \ge c_\Pi$, there exists a set $B\subseteq V(G)$ of size 
$|B|=\frac{c|M|}{\tau-c_\Pi+1}$  such that $\ddstar(G-B)\leq \tau$.
\end{lemma}

\begin{proof}

Let $\tau$ a positive integer with $\tau \ge c_\Pi$, and let us define $B=\{v\in V(G):\ddstar(v)\geq \tau\}$ the set of vertices with ``big'' $\ddstar$ in $G$.
Let us first prove that for any $v\in B$, $|\Nstar(v)\cap M| \ge \ddstar(v)-c_\Pi+1$. Let $E' \subseteq E(G)$ be a maximum matching in $G[\Nstar(v)]$ with $|E'|=\ddstar(v)$. Observe that we cannot have $c_\Pi$ edges $e \in E'$ such that $V(e) \cap M = \emptyset$ as if $v \notin M$, then
vertices of $E'$ together with $v$ would form a $c_{\Pi}$-bundle not hit by $M$, a contradiction, and if $v \in M$, this would contradict the hypothesis $\mu(G[N(v)]-M) < c_\Pi$.
Thus, there is at least $|E'|-c_\Pi+1$ edges of $E'$ intersecting $M$, leading to the desired inequality. 
Thus, we get
\[
|B|\tau\leq \sum_{v\in B} \ddstar(v) \leq \sum_{v\in B}( |\Nstar(v)\cap M|+c_\Pi-1).
\]

Moreover, as for any $v \in V(G)$ there are at most $c$ vertices $u$ such that $v \in \Nstar(u)$, we get $\sum_{v\in B} |\Nstar(v)\cap M| \le c|M|$ by the pigeonhole principle (if the inequality was false, then there would exists $v \in M$ with $|\{u : v \in \Nstar(u)\}| > c$).
This leads to $|B|=\frac{c|M|}{\tau-c_\Pi+1}$. 
\end{proof}

We are now ready to describe the general algorithm to solve $\Pi$.
\begin{proof}[Proof of \autoref{th:subexpASQGM}]
    Given an instance $(G,k)$ of $\Pi$, we 
    first use \autoref{cor:bothbranchings} with $p=k^\epsilon$ to obtain in time $2^{\O(k^{1-\epsilon}\log(k))}$ the set of $2^{\O(k^{1-\epsilon}\log(k))}$ triples $(G_2, M, k_2)$ with   $k_2\leq k$, $|M| = \O(k^{1+\epsilon})$, and $\omega(G_2) \le k^\epsilon$.
   
In order to solve $\Pi$ on $(G,k)$, it is now enough to solve it on these instances $(G_2, k_2)$.
Observe that applying the \autoref{lm:removeBIG} to such $(G_2,k_2,M)$ triple with $\tau \ge c_\Pi$ gives a set $B$ of size at most $\frac{c|M|}{\tau-c_\Pi+1} = \O(\frac{\omega(G_2)^{c_1}k^{1+\epsilon}}{\tau-c_\Pi+1}) = \O(\frac{k^{1+\epsilon+\epsilon c_1}}{\tau-c_\Pi+1})$ such that $G_3=G_2\setminus B$ verifies $\ddstar(G_3)\leq \tau$. 
    
    By assumption on the $ASQGM$ property we then have $\tw(G_3)=\O(P(k^\epsilon, \tau) \grid(G))$. Moreover $\tw(G_2)\leq \tw(G_3)+|B|=\O(P(k^\epsilon, \tau) \grid(G))+\O\left (\frac{k^{1+\epsilon+\epsilon c_1}}{\tau-c_\Pi+1} \right )$ as removing a vertex decreases the treewidth by at most $1$.
    We set $\tau=k^{(c_1+2)\epsilon}$. By assumption we have $P(k^\epsilon, k^{(c_1+2)\epsilon})=\O(k^{\frac 12-\epsilon})$. 
    As $\Pi$ is bidimensionnal, there exists $c_1$ such that if $\grid(G) > c_1 \sqrt{k}$, then $(G,k)$ is a no-instance.
    
    Thus, as $\tw(G_2) = \O\left ( k^{\frac 12-\epsilon}\grid(G)\right )+\O\left (\frac{k^{1+\epsilon+\epsilon c_1}}{\tau}\right )=\O\left ( k^{\frac 12-\epsilon}\grid(G)\right )+\O(k^{1-\epsilon})$, observe that if $\grid(G) \le c_1 \sqrt{k}$, then there exists a constant $c$ such that $\tw(G_2) \le c k^{1-\epsilon}$.
    Thus, we use the treewidth approximation of~\cite{korhonen2022single} on $G_2$ with $\ell = c k^{1-\epsilon}$ to obtain  in $2^{\O(\ell)}n^{\O(1)}$ either a $2\ell+1$ treewidth decomposition, or conclude that $\tw(G_2)>\ell$. In the later case, this implies that $\grid(G) > c_1 \sqrt{k}$, and thus we can conclude that $(G,k)$ is a no instance.
     Otherwise, by definition of problems in $\CPi$ we can solve $\Pi$ in time $\tw^{\O(\tw(G_2))})$, which gives the claimed overall time complexity of 
     $2^{\O(k^{1-\epsilon}\log(k))} \times \tw(G_2)^{\O(\tw(G_2))} = 2^{\O(k^{1-\epsilon}\log(k))}$.
\end{proof}

\subsection{From \texorpdfstring{$\ASQGM(\omega,\lr)$}{ASQGM(ω,lr)} to \texorpdfstring{$\ASQGM(\omega,\ddstar)$}{ASQGM(ω,μ*)}.}\label{subsec:asqgm2}

To be able to use \autoref{th:subexpASQGM}, we need to deal with graph classes that have the $\ASQGM(\omega,\ddstar)$ property. This section provides a general framework for obtaining this property via \emph{local radius}.
The local radius was originally introduced by Lokshtanov et al.~\cite{Lokshtanov23Approx} for disks graphs in the context of approximation algorithms. Here we first extend this definition to string graphs.
To that end, we will see string graphs as graphs admitting  
a \emph{thick representation}. In such a representation every vertex $v$ of the considered graph $G$ corresponds to a subset $\D_v$ of the plane that is homeomorphic to a disk, two intersecting such regions have an intersection with non-empty interior, and
the number of maximal connected regions $\mathbb{R}^2\setminus \bigcup_{v \in V(G)} \partial \D_v$ is finite.

To turn a string representation into a thick one, it simply suffices to thicken each string by a small enough amount so that no new intersections occur. On the other hand, note that any thick representation can be turned into a string representation by replacing each connected subset of the plane $\D_u$ by a string that almost completely fills its interior. 
Note that a thick representation is not necessarily a pseudo-disk representation as here, the intersection of two regions,  $\D_u \cap\D_v$, may not be connected, or it may also be that $\D_u\setminus \D_v$ is not connected.\\
Thick representations allow us to extend the definition of \emph{local radius} to all string graphs.
The next definition is illustrated \autoref{fig:graphBaste}.

\begin{definition}\label{def:arrangement-graphintro}
  Let $G$ be a string graph and $\mS$ be a thick representation of it. Let $\mathcal X$ be the set of 
all maximal connected region $\R$ of $\mathbb{R}^2\setminus \bigcup_{D\in \mS} \partial D$, contained in at least one object of $\mS$.
  We define the arrangement graph of $\mS$, denoted $A_{\mS}$, by:
    \begin{itemize}
        \item adding one vertex of each region of $\mathcal X$, and
        \item adding an edge between two vertices if the boundaries of their regions share a common arc.
    \end{itemize}
    Moreover, for each $v\in G$, we denote $\R_{\mS}(v) \subseteq \mathcal X$ the set of regions included in $\D_v$ (recall that $\D_v$ is the region associated to $v$),
    and $V_{\mS}(v) \subseteq  V(A_{\mS})$ the set of vertices associated to the regions of $\R_{\mS}(v)$ (implying $|V_{\mS}(v)|=|\R_{\mS}(v)|$). Finally, we denote $A_{\mS}(v) = A_{\mS}[V_{\mS}(v)]$.
\end{definition}

\begin{definition}[from~\cite{Lokshtanov23Approx}, extended here to string graphs]
  Let $G$ be a string graph.
  \begin{itemize}
  \item Given a thick representation $\mS$ of $G$,
    \begin{itemize}
      \item for any $v \in V(G)$, we define $\lr_{\mS}(v)$ as the radius of the graph $A_{\mS}(v)$
      \item we define $\lr_{\mS}(G)=\min_{v \in V(G)}\lr_{\mS}(v)$
    \end{itemize}
  \item the local radius $\lr(G)$ of $G$ is the minimum over all thick representation $\mS$ of $G$ of $\lr_{\mS}(G)$.
  \end{itemize}
\end{definition}

     In order to show \ASQGM{} we use the framework of Baste and Thilikos~\cite{baste2022contraction} (originally designed for the classic SQGM property), that we recall now.

	\begin{definition}[contractions \cite{baste2022contraction}]
		Given a non-negative integer $c$, two graphs $H$ and $G$, and a surjection $\sigma~:~V(G)\rightarrow V(H)$ we write $H\leq^c_\sigma G$ if 
		\begin{itemize}
			\item for every $x\in V(H)$, the graph $G[\sigma^{-1}(x)]$ has diameter at most $c$ and
			\item for every $x,y\in V(H),\ xy\in E(H) \iff G[\sigma^{-1}(x)\cup \sigma^{-1}(y)]$ is connected.
		\end{itemize}
		We say that $H$ is a $c$-diameter contraction of $G$ if there is a surjection $\sigma$ such that $H\leq^c_\sigma G$ and we write this $H\leq^c G$. Moreover, if $\sigma$ is such that for every $x \in V(H),~|\sigma^{-1}(x)|\leq c'$, then we say that $H$ is a $c'$-size contraction of $G$, and we write $H\leq^{(c')} G$. If there exists an integer $c$ such that $H\leq^c G$, then we say that $H$ is a contraction of $G$.
	\end{definition}

	\begin{definition}[($c_1,c_2)$-extension \cite{baste2022contraction}]
		Given a class of graph $\mathcal G$ and two non-negative integers $c_1$ and $c_2$, we define the $(c_1,c_2$)-extension of $\mathcal G$, denoted by $\mathcal G^{(c_1,c_2)}$, as the class containing every graph $H$ such that there exist a graph $G\in \mathcal G$ and a graph $J$ that satisfy $G\leq^{(c_1)} J$ and $H\leq^{c_2}J$ (see \autoref{fig:bastefig}).
	\end{definition}

         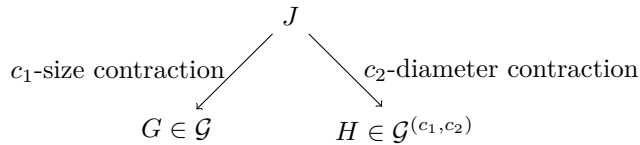
\begin{figure}[!ht]
\centering           
        \begin{tikzpicture}
    		\node (J) at (0,1.5) {$J$};
    		\node (G) at (-1.5,0) {$G\in \mathcal G$};
    		\node (H) at (1.5,0) {$H\in \mathcal G^{(c_1,c_2)}$};
    		
    		\draw[->] (J) -- (G) node[midway,left] {$c_1$-size contraction};
    		\draw[->] (J) -- (H) node[midway, right] {~$c_2$-diameter contraction};
	   \end{tikzpicture}

            \centering

    \caption{A graphical representation of the definition of $\mathcal G^{(c_1,c_2)}$, adapted from \cite{baste2022contraction}.}\label{fig:bastefig}
\end{figure}

	\begin{lemma}[implicit in the proof of {\cite[Theorem~15]{baste2022contraction}}]\label{lem:contract-tw-grid}
		For every integers $c_1,c_2$ and $G\in \mathcal P^{(c_1,c_2)}$, with $\mathcal P$ the class of planar graphs, we have $\tw(G)=\O(c_1 c_2\grid(G))$.
	\end{lemma}

The main result of this section is the following.
\begin{theorem}
\label{lm:ASQGMlr}
String graphs have the $\ASQGM(\omega, \lr)$ property, more precisely for a string graph $G$ we have $\tw(G)=\O( \omega(G) \cdot \lr(G) \cdot \grid(G))$.
\end{theorem}
\begin{proof}

  \begin{figure}[h!]
    \centering
    \includegraphics[scale=.55]{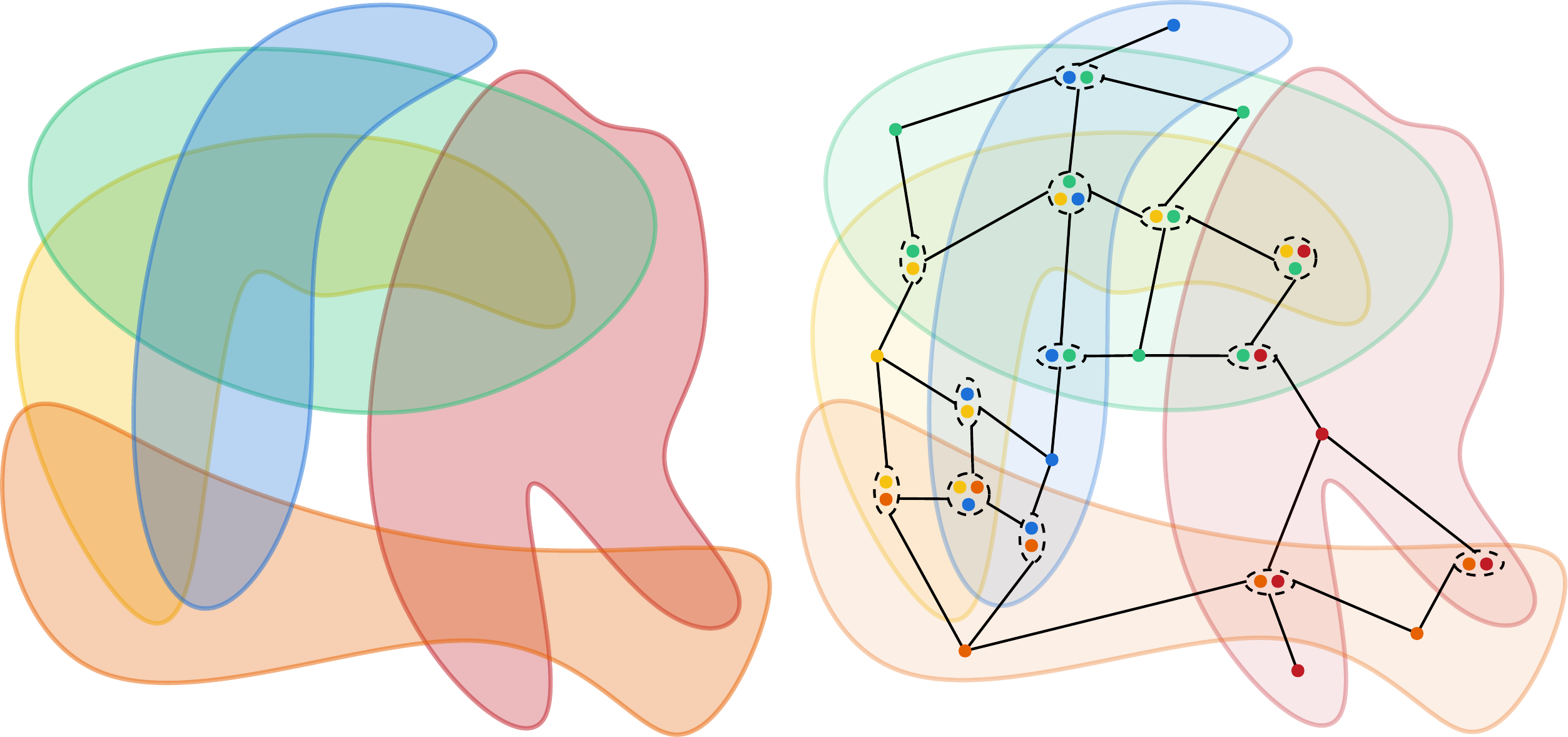}
    \caption{Left: thick representation of a string graph $G$. Right: 
    Illustrates both $A_{\mS}$ and the graph $J$ used in the proof of \autoref{lm:ASQGMlr}. 
    To visualise $A_{\mS}$, consider that each black dotted ellipse is a single vertex (we have $|V(A_{\mS})|=23$). Moreover, if $v$ is the vertex represented in red, we have $|V_\mS(v)|=6$ and $\lr_{\mS}(v)=2$.
    To visualise $J$: for each maximal connected region $\R$ of $\mathbb{R}^2\setminus \bigcup_{D\in \mS} \partial D$, the clique $K_\R$ with more than one vertex is represented by a black dotted ellipse around the clique. For readability only one edge is represented between two cliques instead of the complete bipartite graph.}
    \label{fig:graphBaste}
\end{figure}
  
  Let $G$ be a string graph, and $\mS$ a thick representation such that $\lr_{\mS}(G)=\lr(G)$.
  Let us define a graph $J$ as follows, \autoref{fig:graphBaste} is a representation of the construction. 
  For any maximal connected region $\R$ of $\mathbb{R}^2\setminus \bigcup_{D\in \mS} \partial D$, we add to $J$ a clique $K_\R$ of size $\ply(\R)$.
  Then, for any pair of regions $\{\R_1,\R_2\}$ that share a common arc, we add all edges between $K_{\R_1}$ and $K_{\R_2}$.
  For any $v \in V(G)$, we associate a set $X(v) \subseteq V(J)$ such that
  for any  $\R \in \R_{\mS}(v)$, $|X(v) \cap K_\R|=1$, and such that $X(v) \cap X(u) = \emptyset$ for any $u \neq v$. Notice that the condition $X(v) \cap X(u) = \emptyset$ is possible
  as $|K_\R|= \ply(\R)$, and thus any vertex $v$ can take its ``private'' vertex in $X(v) \cap \R$ for any $\R \in \R_{\mS}(v)$.

  Let us prove that $G$ is a $\lr(G)$-diameter contraction of $J$ by defining a surjection $\sigma: V(J) \rightarrow V(G)$ as follows.
  For any $v \in V(G)$, we define $\sigma^{-1}(v)=X(v)$ (informally we contract all vertices in $X(v)$).
  As for any $v \in V(G)$, $J[X(v)]$ is isomorphic to $A_{\mS}(v)$, we immediately have $\diam(J[\sigma^{-1}(v)])=\lr(G)$.
  Moreover, it is straightforward to check that for every $x,y\in V(G),\ xy\in E(G) \iff J[\sigma^{-1}(x)\cup \sigma^{-1}(y)]$ is connected.
  Now, observe that $A_{\mS}$ (which is planar) is a $\ply(\mS)$-size contraction of $J$ using $\sigma': V(J) \rightarrow V(A_\mS)$ such that for any $v \in V(A_\mS)$, $v$ corresponding to a region $\R$ of the plane delimited by the boundaries of the objects of $\mS$, $\sigma^{'-1}(v)=K_R$.
  As $\ply(\mS) \le \omega(G)$, we get the desired result.
  \end{proof}

The following corollary is immediate from \autoref{th:subexpASQGM} and \autoref{lm:ASQGMlr}. 

\begin{corollary}
  \label{cor:asqgm}
    Given an hereditary graph class $\mathcal C$ which is a subclass of string graphs such that
    \begin{itemize}
        \item maximum clique can be $\O(1)$-approximated in polynomial time,
        \item for any $G \in \mathcal C$, there exists a subneighborhood function $\Nstar$ that has $\O(\omega(G)^{c_1})$-bounded occurrences for some $c_1 \in \mathbb{N}$, and
        \item there exists a multivariate polynomial such that for any $G \in \mathcal C$, $\lr(G)=P(\omega(G), \ddstar(G))$
    \end{itemize}
Then, any problem $\Pi \in \CPi$ admits a parameterized subexponential algorithm on $\C$.
More precisely, let $P'(\omega(G), \ddstar(G))=\omega(G)P(\omega(G), \ddstar(G))$. For any $\epsilon>0$ such that $P'(k^\epsilon, k^{(c_1+2)\epsilon})=\O(k^{\frac 12-\epsilon})$, \FVS can be solved in time $\Ostar(k^{\O(k^{1-\epsilon})})$. This algorithm does not need a representation except if one is required for finding the $\O(1)$-approximation of a maximum clique.
\end{corollary}

\subsection{Upper bounding the local radius for square graphs}
\label{subsec:asqgm3}

Again we provided in the previous section a generic result (\autoref{cor:asqgm}) but so far it might not be clear to the reader which graph classes may satisfy its requirements. To demonstrate the applicability of this result, we show here that square graphs do. This requires to define an appropriate $\Nstar$ and prove that $\lr(G)= \omega(G)^{
O(1)} \cdot \ddstar(G)^{\O(1)}$. A second application, for \CONTACTSEG graphs, is given in the next section.

We say that a graph $G$ is a \emph{square graph} if it is the intersection graph of some collection of (closed) axis-parallel squares in the plane. In the following by \emph{square} we always mean closed and axis-parallel square.
By slightly altering the sizes and positions of the squares in a collection we can obtain a collection where exactly the same pairs of squares intersect and, in addition, all the side lengths of the squares are different from each other and no two sides squares are aligned. Furthermore this can easily be performed in polynomial time.
From now on we will assume that all the representations we consider satisfy this property.

The first requirement of \autoref{cor:asqgm} is provided by following lemma from~\cite{BonnetGM20CliqueDiskLike}, which describes an EPTAS for the clique problem in the more general case of the intersection graph of a fixed convex geometric shape with a central symmetry, while allowing rescaling.

\begin{theorem}[\cite{BonnetGM20CliqueDiskLike}]
    \label{lm:squareCliques}
    There is a polynomial-time 2-approximation of maximum clique in intersection graphs of squares, even when no representation is provided.
\end{theorem}

\begin{definition}\label{def:Nmoins}
  Given a square representation $\mS = \{\D_v\}_{v\in V(G)}$ of a graph $G$, we denote $\ell_{\mS}(\D_v)$ the length of a side of the square $\D_v$,
  $N^-_{\mS}(v)$ (\textit{resp.} $N^+_{\mS}(v)$) the set of vertices $u$ such that $u\in N_G(v)$ and $\ell_{\mS}(\D_u) <  \ell_{\mS}(\D_v)$ (\textit{resp.} $>$). When ${\mS}$ is clear from the context, we will instead write $\ell$, $N^-$ and $N^+$.
\end{definition}

As the lengths of all sides differ, $\{N^+(v), N^-(v)\}$ is a partition of $N(v)$ for every vertex~$v$.

\begin{lemma}
    \label{lm:boundOccSquares}
    Given a square representation $\mS$ of a graph $G$, $N^-$ is a $\O(\omega(G))$-occurrences bounded subneighborhood function.
\end{lemma}
\begin{proof}
    $N^-$ is clearly a subneighborhood  function. For $v\in V(G)$, observe that a square larger than $\D_v$ has to contain one of the four corners of $\D_v$ if the two squares intersect. But a corner of $\D_v$ cannot be contained in more than $\omega(G)$ squares. Hence there are at most $4\omega(G)$ vertices $u\in V(G)$ such that $v\in N^-(u)$, and so $N^-$ is $4\omega(G)$-occurrences bounded.
\end{proof}

We will prove that choosing $N^*=N^-$ allows us to bound the local radius.

\begin{definition}\label{def:HIX}
Given a square graph $G$ with representation $\mS$, 
for any $v\in G$, we define $H(v)$ as a minimum vertex cover of $G[N^-(v)]$, $I(v)=N^-(v) \setminus H(v)$, and $X(v) = H(v) \cup N^+(v)$.
\end{definition}

\begin{claim}\label{claim:HIX}
For every vertex $v$ of a square graph $G$ with representation $\mS$, the following properties hold:
\begin{enumerate}
\item $I(v)$ is an independent set of $G$;
\item $|H(v)| \le 2\ddstar(G)$;
\item $|N^+(v)| = \O(\omega(G))$ (as in the proof of \autoref{lm:boundOccSquares});
\item $|X(v)|=\O(\ddstar(G) + \omega(G))$; and
\item $\{X(v),I(v)\}$ is a partition of $N(v)$.
\end{enumerate}
\end{claim}

\begin{definition}
    For a curve $\mathcal C:[0,1]\rightarrow \mathbb R^2$ such that for $t\in [0,1],~ \mathcal C(t)=(x(t),y(t))$, we say that $\mathcal C$ is \emph{monotonic} if the functions $x$ and $y$ are monotonic. For $k\geq 2$ we say that $\mathcal C$ is $k$-monotonic if it is the composition\footnote{A curve $\mathcal{C}(t) = (x(t), y(t))$ is the \emph{composition} of $k$ curves $(\mathcal{C}_i(t) = (x_i(t), y_i(t)))_{i\in\{1, \dots, k\}}$ if $(x(0), y(0)) = (x_1(0), y_1(0))$, $(x(1), y(1)) = (x_k(1), y_k(1))$,  $(x_i(1), y_i(1)) = (x_{i+1}(0), y_{i+1}(0))$ for every $i\in \{1, \dots, k-1\}$ and the set of points $\{(x(t), y(t)\}, t\in [0,1]\}$ is the union of the $\{(x_i(t), y_i(t)\}, t\in [0,1]\}$ for $i\in \{1, \dots, k\}$.}
    of $k$ monotonic curves.
\end{definition}

\begin{figure}[!ht]
    \centering
    \includegraphics[width=.8\textwidth]{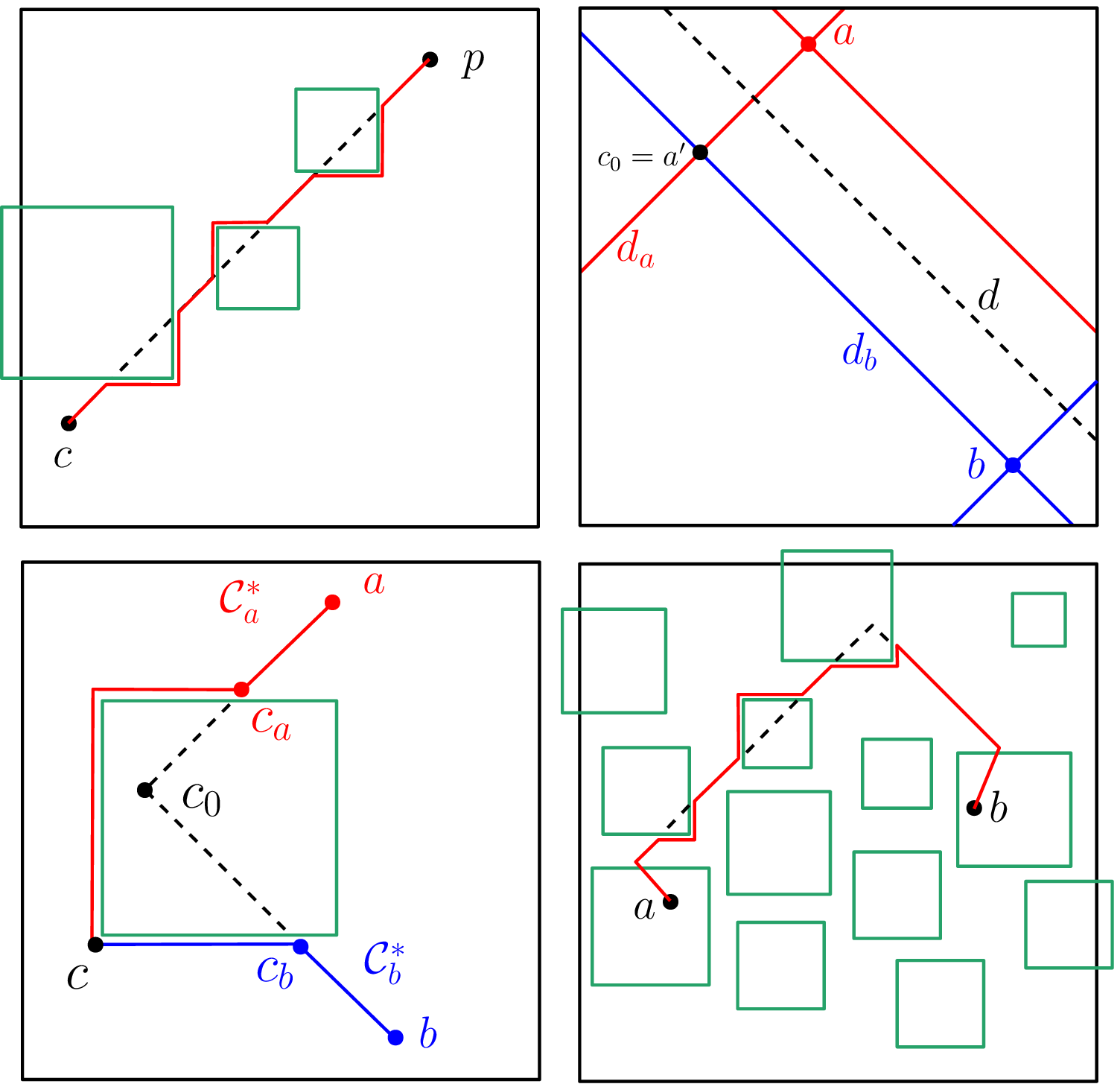}
\caption{Illustrations of the construction used in the proof of the \autoref{lm:curve}. Squares of $I(v)$ are represented in green. Top left: construction used for the \autoref{cl:diag} . Top right: construction used for the \autoref{cl:intersect}. Bottom left: construction used for \autoref{cl:goingaround}. Observe that in this situation $c_a$ and $c_b$ are next to opposite sides of the square containing $c_0$, that $\C_a^*$ can be extended in an counterclockwise direction, and $\C_b^*$ in a clockwise direction, which ensure the existence of a common point $c$ of their monotonic extensions. Bottom right: an example of a $4$-monotonic curve between $a$ and $b$ obtained by the construction of \autoref{lm:curve}. Observe that only two squares of $I(v)$ are crossed.}\label{fig:squarecurve}
\end{figure}

Recall in the next Lemma that $\D_{I(v)}$ denotes the union of all squares in $I(v)$.
\begin{lemma}
  \label{lm:curve}
  Let $G$ be a square graph and $\mS$ a representation.  Let $v\in V(G)$ and $a,b$ two points contained in $\D_v$. There exists a $4$-monotonic curve $\C$  contained in $\D_v$ joining the point $a$ to the point $b$, and crossing at most twice a boundary of the squares of $I(v)$.
\end{lemma}

\begin{proof}
  In what follows, what we call a \emph{diagonal line (resp. half line)} any line (resp. half line) having an angle $+45^\circ$ or $-45^\circ$ with the horizontal axis,
  and a \emph{diagonal of a point $p$ in the plan} a diagonal half line whose endpoint is $p$.

  The first step for the creation of the curve is to reduce to the case where the point $a$ and $b$ are outside $\D_{I(v)}$. If this is not the case, for example if $a$ in contained in a square $s=\D_u$ with $u\in I(v)$, we create a rectilinear curve from $a$ toward the outside of $s$, in a direction such that the intersection of the curve with the boundary of $s$  is still in $\D_v$ (see the construction in \autoref{fig:squarecurve} for an example of such reduction). As such curve is monotonic and crosses the boundary of a square of $I(v)$ exactly once, after the reduction we are in the case where we want to construct a $2$-monotonic curve between two points of $\D_v\setminus \D_{I(v)}$ such that no square of $I(v)$ is crossed. In what follow we suppose we have reduced to this case and we still denote $a$ and $b$ the two points of $\D_v \setminus \D_{I(v)}$ we want to join by a curve.
  \begin{claim}
    \label{cl:diag}
        Given two points $c,p\in \D_v\setminus \D_{I(v)}$ on the same  diagonal line, there is a monotonic curve included in $\D_v\setminus \D_{I(v)}$ between $c$ and $p$.
    \end{claim}
    \begin{proof}
        The construction is represented in \autoref{fig:squarecurve}. The curve is created by starting from the point $c$, then by following the diagonal line toward $p$. When encountering a square $s=\D_u$ of a vertex $u \in I(v)$, it is always possible of getting around $s$ in order to join back the diagonal on the other side, and doing so in a direction such that the curve is still monotonic and contained in $\D_v$.
    \end{proof}

    \begin{claim}
        \label{cl:intersect}
There are diagonals $d_a$ of $a$ and $d_b$ of $b$ intersecting on a point $c_0\in \D_v$.
    \end{claim}
    \begin{proof}
        Consider the line $d$ parallel to the top left to bottom right diagonal of $\D_v$ (see \autoref{fig:squarecurve}), at equal distances of the points $a$ and $b$. By symmetry of the square and of the variables $a$ and $b$, we can suppose that $d$ goes from top left to bottom right, is above the diagonal of $\D_v$, and that $a$ is above $d$. The symmetric $a'$ of the point $a$ relatively to $d$ is inside $\D_v$ and is contained in a diagonal of both $a$ and $b$.
    \end{proof}

Now, if $c_0\in \D_v\setminus \D_{I(v)}$, composing the two curves toward $c_0$ given by the previous claim gives the wanted result.

It remains to deal with the case where 
$c_0$ lies in some square $s=\D_u$ for $u \in I(v)$. Let $c_a$ be a point of $d_a$ between $a$ and the square $s$, at an infinitely small distance outside of $s$. \autoref{cl:diag} gives a monotonic curve $\mathcal C_a^*$ from $a$ to $c_a$. In the same way we define $c_b$ and $\mathcal C_b^*$.

    \begin{claim}
    \label{cl:goingaround}
        There exists a point $c\in \D_v\setminus \D_{I(v)}$ such that $\mathcal C_a^*$ and $\mathcal C_b^*$ can be extended to $c$ while still being monotonic and contained in $\D_v\setminus \D_{I(v)}$.
    \end{claim}

    \begin{proof}
        
        We can assume that $d_a$ and $d_b$ are perpendicular as otherwise the points $a$ and $b$ are on the same diagonal and so \autoref{cl:diag} gives the wanted result by taking $c=b$. 
        Observe that if $c_a$ and $c_b$ are arbitrarily close to the same side of $s$, then prolonging $\C_a^*$ toward $c_b$ would keep the curve monotonic, as $\C_a^*$ was already going toward $d_b$ as $d_a$ and $d_b$ intersect in $s$. So taking $c=c_b$ would give the wanted result.

        Otherwise if $c_a$ and $c_b$ are at arbitrarily small distance from two different sides, observe that the curve $\C_a^*$ can be extended running alongside the boundary of $s$ until crossing $2$ corners. The same is true for $\C_b^*$ so the only situation where those extensions do not cross each other would be if $c_a$ and $c_b$ are next to opposite side of $s$, and that the orientations of $d_a$ and $d_b$ force the extensions of $\C_a*$ and $\C_b^*$ to go in the same direction around $s$. However, this is impossible: as $d_a$ and $d_b$ cross each other inside of $s$, one extension will go clockwise around $s$ and the other counterclockwise (see \autoref{fig:squarecurve}). This ensures that $\C_a^*$ and $\C_b^*$ can be extended around $s$ while still being monotonic in order for them to join on a point $c$ while staying outside of $\D_{I(v)}.$
    \end{proof}
    
    Composing the two curves obtained by the above claim gives a path as wanted.
\end{proof}

  \begin{figure}[!ht]
    \centering
    \includegraphics[width=.8\textwidth]{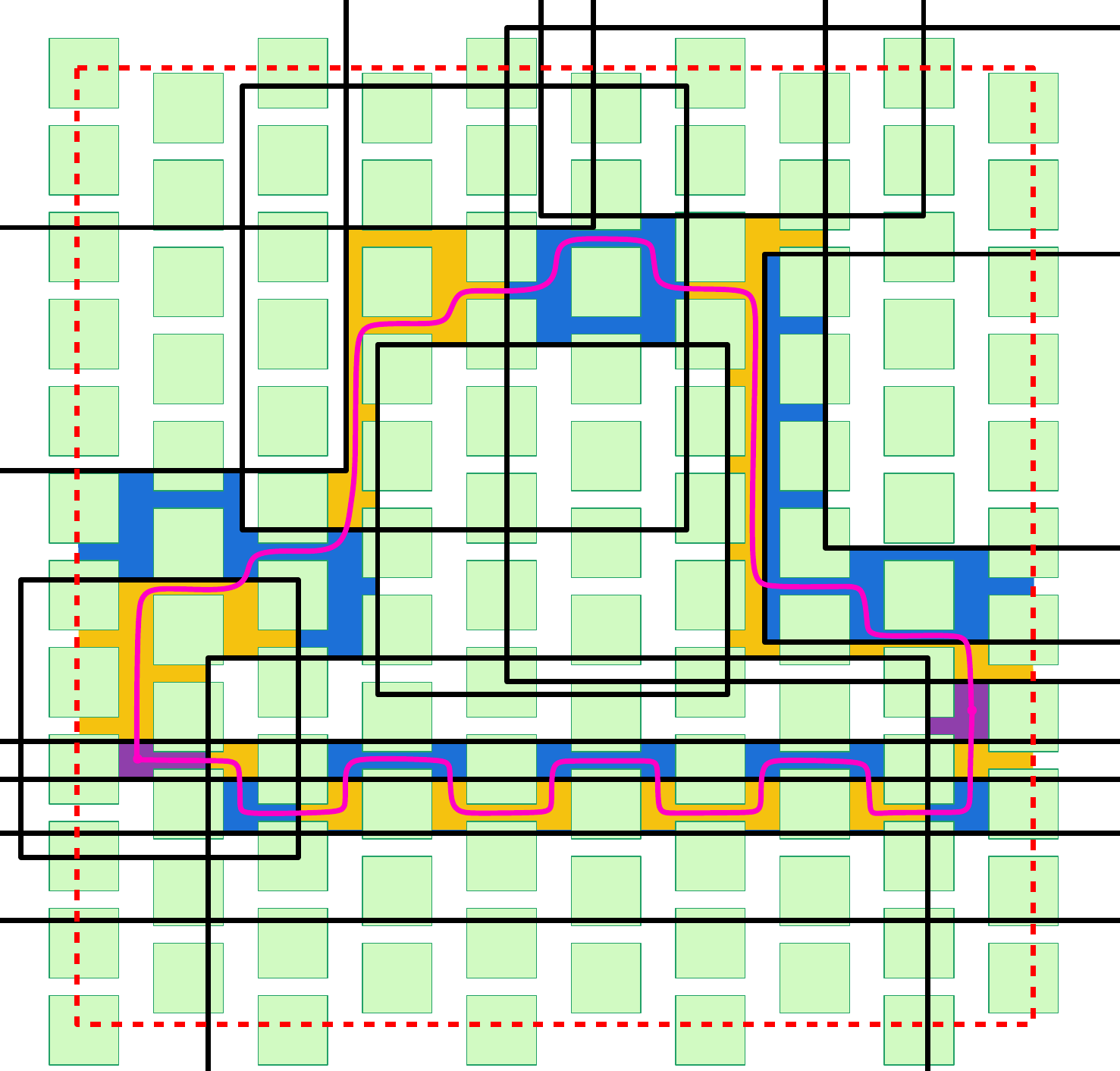}
\caption{Examples of paths in the configuration graph, with $\D_v$ represented with a dashed red square, $I(v)$ by green squares and the sides of the squares of $X(v)$ in black. Here we can see two curves between the two purple regions, $\mathcal{C}_1$ (that goes up and then down) and $\mathcal{C}_2$, and the path in $A_{\mS}(v)$ associated to each curve as in the proof of \autoref{lm:lrBoundSquare}, where the regions traversed by the paths are alternatively colored blue and yellow. Notice that $\mathcal{C}_1$ is $2$-monotone, whereas $\mathcal{C}_2$ is $c$-monotone, where $c$ could be made arbitrary large by creating more and smaller squares in $I(v)$. As $c$ is large, there is a side of a square in $X(v)$ crossed many times (eight) by $\mathcal{C}_2$, and thus we do not use curve like $\mathcal{C}_2$ in the proof.}
\label{fig:pathsquare}
\end{figure}

We are now ready to prove the main combinatorial statement of this section.
\begin{lemma}
    \label{lm:lrBoundSquare}
Let $G$ be a square graph. There exists a subneighborhood function $\Nstar$ which is $\omega(G)$-occurrences bounded and such that $\lr(G)=\O(\ddstar(G) + \omega(G))$.
\end{lemma}
\begin{proof}
  Let $\mS$ be a square representation of $G$, and let $\Nstar$ as defined in \autoref{def:Nmoins}, which is $\omega(G)$-occurrences bounded according to \autoref{lm:boundOccSquares}.
   Let us now prove that $\lr_\mS(G)=\O(|X(v)|)$. This will imply the required result as $\lr(G) \le \lr_\mS(G)$ and $|X(v)|=\O(\ddstar(G) + \omega(G))$ by \autoref{claim:HIX}.
    To that end, let us bound the diameter of $A_{\mS}[V_\mS(v)]$. Let $u,v$ be two vertices of $A_{\mS}[V_\mS(v)]$, and let us bound the distance between these two vertices. 
    Remember that any vertex in $A_{\mS}[V_\mS(v)]$ corresponds to an inclusion-wise maximal rectangular region of the plane included in $\D_v$, and delimited by edges of squares of $\mathcal{S}$. 
    Let $a$ and $b$ be points in the regions of $u$ and $v$ respectively.
    Notice that to any curve inside $\D_v$ we can associate a path in $A_{\mS}[V_\mS(v)]$ by considering the sequence of regions visited by $\mathcal{C}$, and associate to each of the region its corresponding vertex in $A_{\mS}[V_\mS(v)]$ (see \autoref{fig:pathsquare}).
    Thus, we will upper bound the distance from $u$ to $v$ in $A_{\mS}[V_\mS(v)]$ by constructing a curve $\mathcal{C}$ from $a$ to $b$, and by counting the length of the sequence of regions visited by $\mathcal{C}$.

    We use for $\mathcal C$ the $4$-monotonic curve between $a$ and $b$ defined in \autoref{lm:curve}.
    Observe the following property $\pi_0$: any monotonic curve inside $\D_v$ crosses at most $4|X(v)|$ sides of squares in $X(v)$.
    Indeed, as each square in $X(v)$ has at most $4$ sides intersecting $\D_v$, and any side, as a vertical or horizontal segment intersecting in $\D_v$, can be crossed at most one time by a monotonic curve.
    Observe also that, each time $\mathcal{C}$ leaves its current region, $\mathcal{C}$ must cross a side of a square in $N(v)$.
    However, the total number of crossings between $\mathcal{C}$ and a side of a square in $N(v)$ is at most $16|X(v)|+4$, as each of the four monotonic part of $\mathcal{C}$ crosses at most $4|X(v)|$ sides of squares in $X(v)$ (by $\pi_0$), 
    and $\mathcal{C}$ crosses at most $4$ sides of squares in $I(v)$ (the worst case being when $a \neq a'$, and $\mathcal{C}_{a\rightarrow a'}$ crosses the corner of the square in $I(v)$ containing $a$, and same for $b,b'$).
    Thus, the curve $\mathcal{C}$ goes from a region to the next one at most $16|X(v)|+4$ times, implying that the diameter of $A_{\mS}[V_\mS(v)]$, and so the local radius $\lr_{\mS}(G)$, are in $\O(|X(v)|)$.
\end{proof}

As announced in the introduction of the section, we are now able to apply 
\autoref{cor:asqgm}.

\begin{theorem}\label{thm:pisquare}
    Any problem $\Pi \in \CPi$ can be solved in time $2^{\O(k^{9/10} \log (k))}n^{\O(1)}$ in square graphs, even when no representation is given.
\end{theorem}

\begin{proof}
Let $\Pi \in \CPi$.
According to \autoref{lm:squareCliques}, \autoref{lm:lrBoundSquare}, we can apply \autoref{cor:asqgm} with $c_1=1$, and $P(x,y)=x+y$.
  This implies that for any $\epsilon$ such that $k^\epsilon(k^\epsilon+k^{3\epsilon})=\O(k^{\frac{1}{2}-\epsilon})$, $\Pi$ can be solved in $\Ostar(k^{\O(k^{1-\epsilon})})$ in square graphs.
  Taking $\epsilon=\frac{1}{10}$ leads to the claimed complexity.
\end{proof}

\subsection{Upper bounding the local radius for contact segment graphs}
\label{sec:asqgm-contact-seg}

As announced above, we will show here that \CONTACTSEG graphs satisfy the requirements of \autoref{cor:asqgm}, which will have the following algorithmic consequence.

\begin{theorem}\label{thm:fvscontact}
    Any problem $\Pi \in \CPi$ can be solved in time $2^{\O(k^{7/8} \log (k))}n^{\O(1)}$ in \CONTACTSEG graphs, even when no representation is given.
\end{theorem}

Again, the first requirement of \autoref{cor:asqgm} is the approximability of maximum clique. This is \autoref{lm:segmentCliques}, for the proof of which we first prove two simple statements.

\begin{figure}
    \centering
    \includegraphics{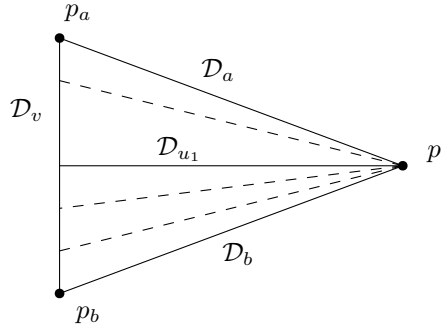}
    \caption{The construction used in the proof of \autoref{lm:cliquePoint}. The dashed segments represent examples of potential additional segments in the representation of the clique.}
    \label{fig:proofPoint}
\end{figure}

\begin{lemma}
    \label{lm:cliquePoint}
    Let $n\geq 1$. For every representation of $K_n$ as a contact-segment graph there is a point of the plane that belongs to at least $n-1$ segments.
\end{lemma}
\begin{proof}
    The result is trivial for $n\leq 3$. For $n\geq 4$, let $\mS$ a representation of $K_n$ as contact-segment. If for all segments, the contact with the other segments is made on a single point, then this contact point is common for all the segments and we get the desired result. Otherwise, as represented in \autoref{fig:proofPoint} let $\D_v$ a segment in contact with the other segments on at least two distinct points $p_a,p_b\in \D_v$ and $a,~b$ vertices such that $\D_a$ (respectively $\D_b$) intersects $\D_v$ on $p_a$ (respectively $p_b$). As they correspond to adjacent vertices, $\D_a$ and $\D_b$ have to intersect: we denote $p$ their contact point. By definition of contact-segments graphs, $p\notin \D_v$. Let $u_1$ a vertex distinct from $v,~ a$ and $b$. $\D_{u_1}$ is in contact with $\D_v, ~\D_a$ and $\D_b$, this can be done only by containing either $p_a,~ p_b$ or $p$. Without loss of generality suppose that $\D_{u_1}$ contains $p$, and so an endpoint of the segment $\D_{u_1}$ is contained in $\D_v$. For the remaining vertices $u_2,\dots,u_{n-3}$, we still have that either $p$, $p_a$ or $p_b$ are contained in their segment, but as it is not possible to cross the segment $\D_{u_1}$, it is always $p$ which is contained in the segment. So all the segments except $\D_v$ intersect on $p$, which yields the claimed result.
\end{proof}
\begin{lemma}
    \label{lm:polyMaxClique}
    A contact-segment graph $G$ with $n$ vertices has $\O(n^2)$ maximal cliques.
\end{lemma}
\begin{proof}
    There are $\O(n^2)$ maximal cliques with at most $2$ vertices in $G$, and for a maximal clique $C$ of size at least $|C|\geq 3$, \autoref{lm:cliquePoint} gives the existence of a point $p_C$ contained in all the segments of the representation of the maximal clique except at most one. This point is necessarily the endpoint of a segment, so $p_C$ can take at most $2n$ distinct values. 
    
    Let us now distinguish two cases. If all the segments representing $C$ contain $p_C$: because $C$ is maximal, it is exactly the segments containing $p_C$ and so is entirely determined by $p_C$. This gives that there are at most $\O(n)$ maximal clique for this case.
    The second case is when there is some vertex $v\in C$ such that $p_C$ is not contained in $\D_v$. By denoting $S=\{u\in V(G),~ p_C\in \D_u\}$ we have $C=(S\cap N(v))\cup \{v\}$ as $C$ is maximal. So this time $C$ is determined by the point $p_c$ and the vertex $v$ not containing $p_C$. So there are at most $\O(n^2)$ maximal clique for this case.
    In total we obtain at most $\O(n^2)$ maximal cliques in $G$.
\end{proof}

\begin{lemma}
    \label{lm:segmentCliques}
    There is a polynomial time algorithm that, given a contact-segment graph, returns a maximum clique, even if no representation is provided.
\end{lemma}
\begin{proof}
    By \autoref{lm:polyMaxClique} contact-segment graphs have a polynomial number of maximal cliques.
    Those can be enumerated with polynomial delay (for instance using \cite{enumerateMaxCliques}) so overall in polynomial time one can enumerate all maximal cliques of a contact-segment graph, and return one with the maximum number of vertices.
\end{proof}

\begin{definition}[interior and non-trivial contact points]
Given a \CONTACTSEG representation of a graph and a point $p$ in the plane, we denote by $\E^p$ the set of segments whose one endpoint is $p$.
Given a segment $s$, we denote by
\begin{itemize}
\item $\E(s)$ the two endpoints of $s$,
\item $\ICP(s)$ (for interior contact point) the set of interior points of $s$ that are endpoints of other segments,
\item $\NT(s)$ (for non-trivial) the subset of those points $p \in \ICP(s)$ for which there exist $s_1$ and $s_2$ on the two sides of $s$ such that $p \in \E(s_1) \cap \E(s_2)$, and
\item $\T(s) = \ICP(s)\setminus \NT(s)$ ($\T$ for trivial).
\end{itemize}
\end{definition}

\begin{figure}
    \centering
    \includegraphics[scale=.6]{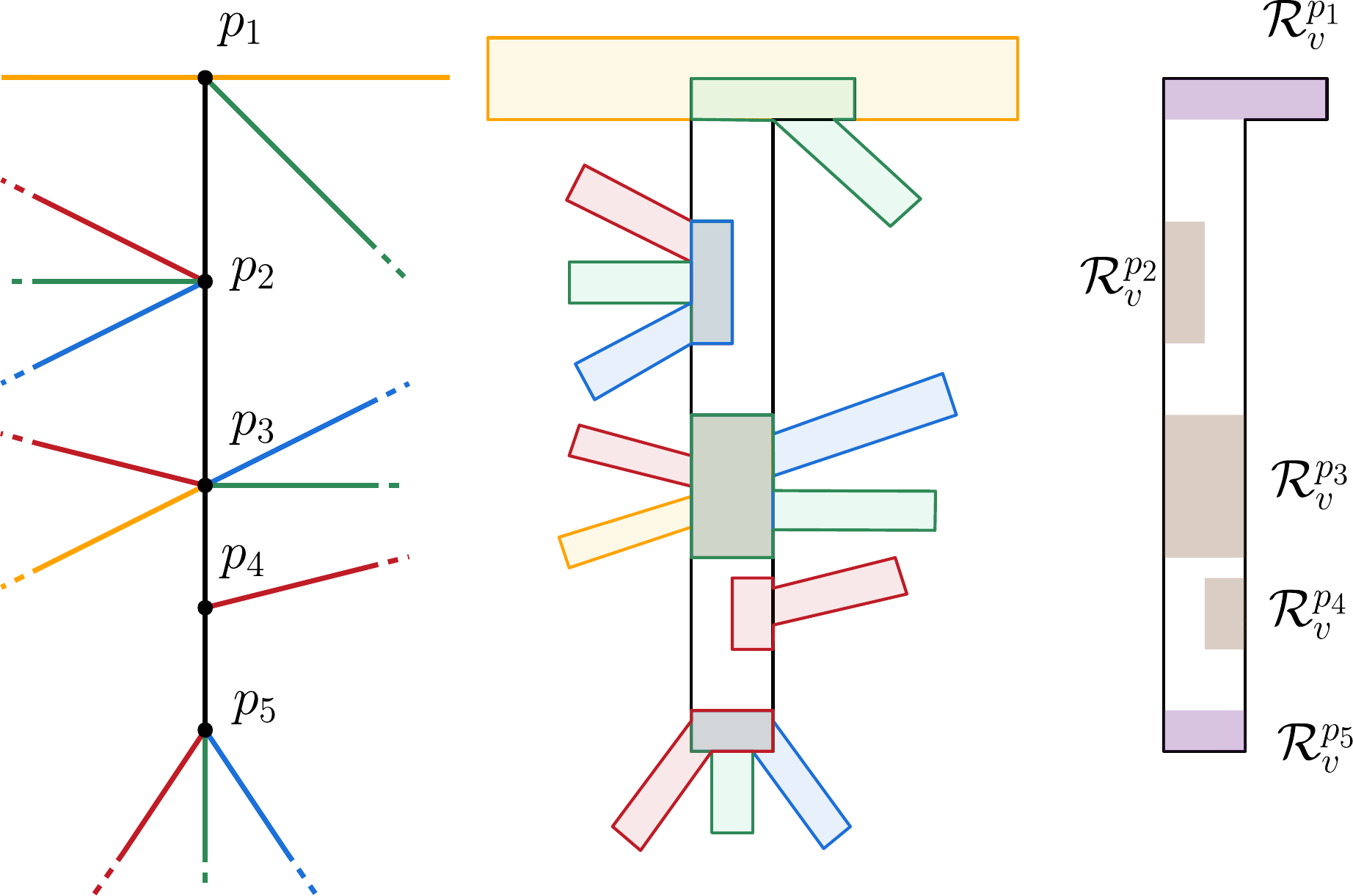}
    \caption{Left: a representation $\mS_1$ of a contact-segment graph. Middle: the thick representation $\mS_2$ defined in \autoref{lm:lrBoundSegment}. Right: the slots of the region associated to a vertex $v$. Here by denoting $s$ the segment representing $v$ we have $\E(s)=\{p_1,p_5\}$, $\ICP(s)=\{p_2,p_3,p_4\}$, $\NT(s)=\{p_3\}$ and $\T(s)=\{p_2,p_4\}$.}
    \label{fig:thickcontactseg}
\end{figure}

\begin{definition}\label{def:nstar2}
  Given a \CONTACTSEG graph $G$ and a representation $\mS= \{\D_v \mid v\in V(G)\}$, for any $v$ we define $\NSstar(v)=\{u \in N(v) \mid \E(\D_u) \cap \NT(v) \neq \emptyset\}$.
When $\mS$ is clear from context, we use $\Nstar$ instead of $\NSstar$.
\end{definition}
Informally, $\Nstar(v)$ is the set of all vertices intersecting $v$ in a non-trivial contact point.

\begin{remark}\label{rem:nstar2}
 With the same definitions as in \autoref{def:nstar2}, the following holds: 
  \begin{itemize}
  \item $\Nstar$ is a $2$-occurrences bounded subneighborhood function (direct from the definition).
  \item  For any $v \in V(G)$, $|\NT(v)|\le \ddstar(v)$ (as all segments ending in a point $p \in \NT(s)$ form a clique of size at least $2$, and thus contribute to a least $1$ on the matching in $G[\Nstar(v)]$).
  \end{itemize}
\end{remark}

\begin{lemma}
  \label{lm:lrBoundSegment}
 Let $G$ be a \CONTACTSEG graph. There exists a subneighborhood function $\Nstar$ which is $2$-occurences bounded and such that $\lr(G)=\O(\ddstar(G))$.
\end{lemma}
\begin{proof}
  Let $\mS_1$ be a \CONTACTSEG representation of $G$, and $\Nstar$ as defined in \autoref{def:nstar2}.
  Let us now define a thick representation $\mS_2 = \{\D_v \mid v\in V(G)\}$ of $G$ such that $\lr_{\mS_2}(G) = \O(\ddstar)$. This will imply $\lr(G) \le \lr_{\mS_2}(G)=\O(\ddstar(G))$.
  Informally, we will give to each segment some thickness, and associate special ``slots'' to draw intersection with other segments (see \autoref{fig:thickcontactseg}).

 To any vertex $v \in V(G)$ corresponding to a segment $s$ in $\mS_1$, we associate a region $\D_v$, and $|\ICP(s)|+2$ disjoint and connected subregions (called \emph{slots}) 
 $\R^p_v$ for $p \in \E(s) \cup \ICP(s)$ such that in particular:
 
  \begin{itemize}
  \item $\{v,v'\} \in E(G)$ iff $\D_v \cap \D_{v'} \neq \emptyset$ (i.e. $\mS_2$ is indeed a representation of $G$)
  \item intersections occur within slots: for any vertex $v' \in N(v)$ corresponding to a segment $s'$ in $\mS_1$, with $p = s' \cap s$, we have $\D_{v'} \cap \D_v  = \R^p_v = \R^p_{v'}$.
   \item for any $p \in \E(s) \cup \T(s)$,  $\D_v \setminus \R^p_v$ is connected, 
  \item for any $p \in \NT(s)$, $\D_v \setminus \R^p_v$ is a set of two connected regions.
  \end{itemize}
 
  This completes the description of $\mS_2$. 
  For any $v \in V(G)$ corresponding to a segment $s$ in $\mS_1$, we have $\lr_{\mS_2}(v) = \O(|\NT(s)|)$ (as for any $p \in \T(s)$, $\D_v \setminus \R^p_v$ is connected), and thus $\lr_{\mS_2}(v)=\O(\ddstar(v))=\O(\ddstar(G))$ by \autoref{rem:nstar2}.
  
\end{proof}

We are now ready to prove the main algorithmic result of the section.
\begin{proof}[Proof of \autoref{thm:fvscontact}]
Let $\Pi \in \CPi$. According to \autoref{lm:segmentCliques} and \autoref{lm:lrBoundSegment}, we can apply \autoref{cor:asqgm} with $c_1=0$, and $P(x,y)=y$.
  This implies that for any $\epsilon$ such that $k^{3\epsilon}=\O(k^{\frac{1}{2}-\epsilon})$, $\Pi$ can be solved in time $\Ostar(k^{\O(k^{1-\epsilon})})$ in \CONTACTSEG.
  Taking $\epsilon=\frac{1}{8}$ leads to the claimed complexity.
\end{proof}

\section{\ETH based hardness results}
\label{sec:negative}

Let us first start with the following result on \TH and \OCT.

 \begin{theorem}\label{th:lbk3hit}
    Under the ETH, \TH and \OCT cannot be solved in time $2^{o(n)}$ on $n$-vertex \DEUXDIR{} graphs.
 \end{theorem}

Before proceeding to the proof of \autoref{th:lbk3hit} we need to introduce some definitions about the gadgets used in our reduction.

	\begin{definition}
   	For $k\geq 2$, a \emph{$k$-polygon} $P$ is a \DEUXDIR{} graph composed of $4k$ axis-parallel segments in the plane such that $V(P)=H\cup V\cup C$ with :
   	\begin{enumerate}
   		\item $H$ is a set of $k$ disjoint horizontal segments of non-zero length;
     \item $V$ is a set of $k$ disjoint vertical segments of non-zero length;
     \item every segment of $H$ intersects exactly two segments of $V$, and vice-versa;
    \item $C$ consists of zero-length segments located at each intersection point between a segment of $H$ and a segment of $V$; and
    \item the intersection graph of the segments in $P$ is connected.
   	\end{enumerate}
   \end{definition}
   Notice that in the definition above, $|C|=2k$.
   See \autoref{fig:fPoly} for a depiction of a 3-polygon.
   
   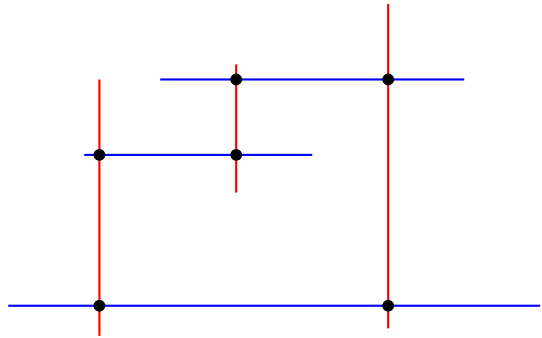
\begin{figure}[!h]
   	\begin{center}
   		\begin{tikzpicture}
   			\draw[draw=blue, thick] (0,1) -- (7,1);
   			\draw[draw=blue, thick] (1,3) -- (4,3);
   			\draw[draw=blue, thick] (2,4) -- (6,4);
   			
   			\draw[draw=red, thick] (1.2,0.6) -- (1.2,4);
   			\draw[draw=red, thick] (3,2.5) -- (3,4.2);
   			\draw[draw=red, thick] (5,5) -- (5,.7);
   			
   			\filldraw[black] (1.2,1) circle (2pt) node{};
   			\filldraw[black] (1.2,3) circle (2pt) node{};
   			\filldraw[black] (5,1) circle (2pt) node{};
   			\filldraw[black] (5,4) circle (2pt) node{};
   			\filldraw[black] (3,4) circle (2pt) node{};
   			\filldraw[black] (3,3) circle (2pt) node{};
   		\end{tikzpicture}
   		\caption{An example of a $3$-polygon, where the segments in $H$ are blue, the ones in $V$ are red, and the length $0$ segments in $C$ are represented by black dots.}
     \label{fig:fPoly}
   	\end{center}
   \end{figure}
   
   \begin{lemma}
   \label{lem:polygon}
   	Let $P$ be a $k$-polygon. Then $P$ does not have any triangle hitting set of size $k-1$, and has exactly two triangle hitting sets of size $k$: the non-zero horizontal segments and the non-zero vertical segments.
   \end{lemma}
    \begin{proof}
    Let $P$ be a $k$-polygon, and let $\{H,V,C\}$ be the partition of $V(P)$ named as in the definition of $k$-polygon. Notice that every vertex of $C$ together with its two neighbors (one vertical, one horizontal, by definition) forms a triangle, and that every triangle $P$ has this form. So $P$ contains $2k$ triangles.
    As every segment is part of at most two triangles, any triangle hitting set has size at least $k$, which shows the first part of the statement.
    Suppose now that $X$ is a triangle hitting set of size $k$. In order to intersect the $2k$ triangles of $P$, each of the $k$ segments in $X$ have to take part in two triangles, so $X\cap C=\emptyset$, and no triangle can be hit twice, so $S$ is an independent set. The induced subgraph over $V\cup H$ is a cycle of size $2k$, whose independent sets of size $k$ are $H$ and $V$. So $X=V$ or $X=H$, as claimed.
    \end{proof}

   \begin{proof}[Proof or \autoref{th:lbk3hit}]
 	The proof is a reduction from \THREESAT.
   
   Let $\varphi$ be a $3$-SAT instance with $n$ variables $x_1,\dots,x_n$ and $m$ clauses $C_1,\dots, C_m$. Without loss of generality we may assume that each variable appears in some clause and that $\varphi$ has no clause with only one literal (otherwise it could be easily simplified). For our reduction we also want to avoid clauses with $3$ literals all positive or all negative. To do so, for any clause of the form $x_a\lor x_b \lor x_c$ we define an additional variable $y_i$ and we replace the clause by the equivalent clauses $x_a\lor x_b \lor \overline{y_i}$ and $x_c\lor y_i$, and similarly for clauses containing only negative literals. Notice that after performing these operations the number of clauses and variables increased by $\O(m)$.

   Let us now construct a \DEUXDIR{} graph $G$ from the formula $\varphi$.
   In this graph each variable $x_i$ is represented by a $k_i$-polygon with $k_i$ to be specified later. (We only describe the non-zero segments, as the position of the zero-length segments is uniquely determined by those, by definition.)
   To every clause $C$, we associate a point of the plane $z_C$. We construct the variable polygons such that, for every clause $C$ and variable $x$, each of the following is satisfied:
   \begin{enumerate}
       \item if the literal $x$ (respectively $\bar{x}$) appears in $C$, then some vertical (respectively horizontal) segment of the polygon of the variable $x$ ends at $z_C$;
       \item if $C$ contains only two literals, then a new vertex with zero-length segment is added at position $z_C$;
       \item if $x$ does not appear in $C$, then $z_C$ does not belongs to any segment of the polygon of~$x$; and
       \item if two horizontal (respectively vertical) segments intersect, their intersection consists in a unique point of the form $z_{C'}$ for some clause $C'$.
   \end{enumerate}
   
    Also, we want the number of segments in each polygon  and its number of intersections with each other polygon to be linearly bounded from above by the number of clauses the corresponding variable appear in. Such a configuration can for instance be obtained by initially drawing the polygons of the variables as concentric rectangles and then, for every clause $C$, picking a point $z_C$ outside of the outermost triangle and connecting corresponding the polygons to it. See \autoref{fig:fReducbis} for an illustrative example and \autoref{fig:fReduc2} for a depiction of the connection to the $z_C$'s.

   \begin{figure}
        \centering
       	\includegraphics[scale=.9]{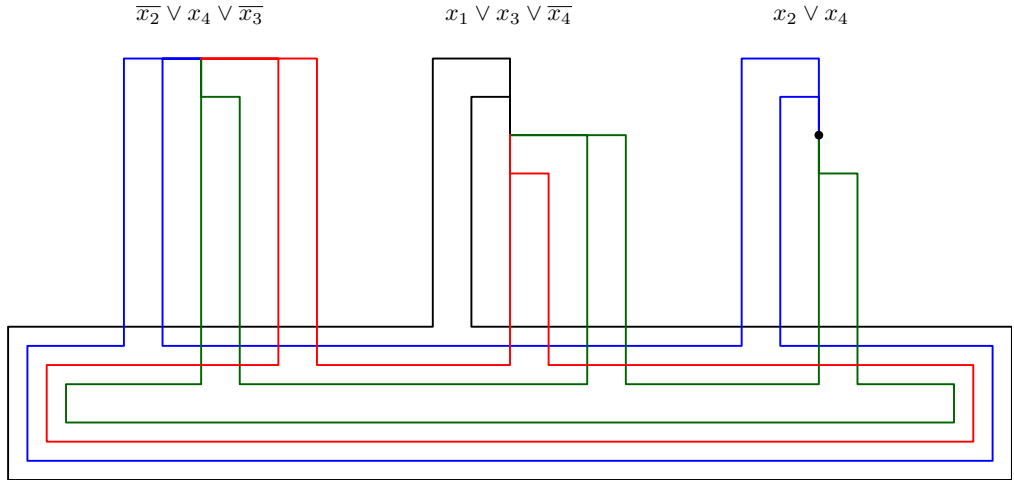}
   	\caption{The construction for the formula $(\overline{x_2}\lor x_4 \lor \overline{x_3})\land (x_1\lor x_3 \lor \overline{x_4})\land (x_2\lor x_4)$. The zero-length segments at each corner of the $k$-polygons are not represented, while that added for the clause with two variables is depicted with a black dot.}
    \label{fig:fReducbis}
   \end{figure}

  \begin{figure}
    \centering
    \includegraphics[]{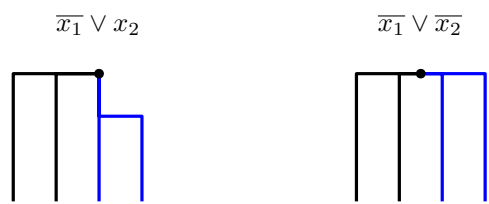}
    \caption{The two remaining possibilities for the clause gadget. The length $0$ segments at each corner of the $k$-polygons are not represented. The black dots represent the length $0$ segments we add for clauses with $2$ literals. }
    \label{fig:fReduc2}
  \end{figure}
	 For $1\leq i \leq n$ we define $k_i$ such that the variable gadget for the variable $x_i$ constructed as defined above is a $k_i$-polygon.
  Let $k = \sum_{i=1}^n k_i$. We define a \emph{good triangle hitting set} as a triangle hitting set of size $k$.
  \begin{claim}\label{cl:goodTHS}
	    If a good triangle hitting set exists, it is a minimum triangle hitting set, and a minimum odd cycle transversal set, moreover it only contains non-zero length segments and all the segments from a polygon it contains have the same orientation.
	\end{claim}
    \begin{proof}
        Let $S$ be a good triangle hitting set of $G$. By \autoref{lem:polygon}, for each variable $x_i$, its corresponding $k_i$-polygon needs $k_i$ segments in order to hit all its triangles.
        As $k=\sum_{i=1}^n k_i$ and no segment is part of two distinct polygons, we can conclude that exactly $k_i$ segments of the considered polygon are in $S$, and $S$ is a minimum triangle hitting set.
        Again by \autoref{lem:polygon} we get that those $k_i$ segments have non-zero length and either consist of all vertical segments of the polygon, or of all its horizontal segments. Moreover, this implies that $S$ does not contains the zero length segments added for clauses with two literals.
        It remains to prove that $S$ is a minimal odd cycle transversal set. Minimality is obtained by the minimality of $S$ as triangle hitting set and observing that a odd cycle transversal is a triangle hitting set. Let us consider a cycle in $G$ with odd size and show that $S$ contains at least one segment of this cycle. To do so it is enough to prove the cycle necessarily contains $2$ segments of a same corner, as in this case it would contains the two non-zero segments of this corner, one of which necessarily belongs to $S$. By contradiction suppose that no pair of segments in the cycle intersect in a corner. Then the intersection between two adjacent segments of the cycle is done in their interior (i.e., not at their endpoints) and so they do not have the same orientation: one is horizontal, the other is vertical. Hence the segments of the cycle are alternatively horizontal and vertical and the cycle has even size, a contradiction. \cqed
     \end{proof}

	\begin{claim}\label{cl:goodiffsat}
	    $G$ has a good triangle hitting set if and only if $\varphi$ is satisfiable.
	\end{claim}
    \begin{proof}
	Direction ``$\Rightarrow$''. Let $S$ be a good triangle hitting set of $G$.
	We define a truth assignment for the variables by setting the variable $x_i$ at false if the segments of the corresponding polygon being in $S$ are the horizontal segments, and true if they are the vertical ones. For a clause $C$ of $\varphi$, let $\Delta_C$ denote the set of the 3 vertices containing~$z_C$ (which trivially induce a triangle).
    Because $S$ is a hitting set and by \autoref{cl:goodTHS}, we have $S\cap \Delta_C\neq \emptyset$. Let $x_i$ be a variable such that its $k_i$-polygon contains a segment in $S\cap \Delta_C$. Suppose this segment is vertical, this means that the variable $x_i$ appears in $C$ in a positive literal, and that in the constructed assignment $x_i$ is true, and the clause $C$ is verified. We can obtain the same conclusion is the case of a vertical segment, so each clause is verified with the defined truth assignment, and so $\varphi$ is satisfiable.
 
 Direction ``$\Leftarrow$''. For an assignment of the variables $x_i$ such that $\varphi$ is verified, we construct a set $S$ of size $k$ with the same encoding than before: $S$ contains the vertical segments of the $k_i$-polygon for the variables $x_i$ assigned to true, and the horizontal ones otherwise.
 Let us now show that $S$ is a hitting set.
 
 All the triangles contained in the polygons encoding the variables are hit, as they always contain an horizontal segment and a vertical segment of the same polygon. The remaining triangles are the one formed by each clause $C$ on the point $z_C$. As $\varphi$ is verified by the considered truth assignment, there is a literal which evaluate to true. Let $x_i$ be the variable contained in this literal. If the literal is positive, then $x_i$ was assigned to true, so all the horizontal segments of its polygon are in $S$; and the segment in contact with $z_C$ is horizontal. So the triangle at the point $z_C$ is hit by $S$. The case of a negative literal have the same conclusion. So all triangles of $G$ are hit by $S$.\cqed
\end{proof}	

\begin{claim}
$|G| = \O(m)$.
\end{claim}
\begin{proof}
	   As argued above for every variable, the number of segments of its polygon is linearly bounded by the number of clauses it appears in. Each clause contains at most $3$ variables and there are at most $m$ extra vertices (for clauses of size 2) so the total number of segments is $\O(m)$.
    \cqed
    \end{proof}

    We described the construction of a graph $G$ and integer $k$ from $\varphi$ such that $(G,k)$ is a positive instance of \TH or \OCT if and only if $\varphi$ is satisfiable.
    Besides, the construction can clearly be done in time polynomial in $n+m$ and the graph has $n' = \O(m)$ vertices. Therefore any algorithm solving \TH{} or \OCT{} in time $2^{o(n')}$ for input graphs of size $n'$ could be used to solve \THREESAT{} in time $2^{o(m)}$ on formulae with $m$ clauses, which would refute the \ETH.
 \end{proof}

 We observe that the instances constructed as in \autoref{fig:fReducbis} contain large complete bicliques.
One way to forbid such bicliques is to restrict the number of crossings in our constructions. We do so by two different approaches. The lower bounds we obtain are weaker in terms of value, however they apply more generally as they hold for more restricted graph classes.
The bound in the following result provides an interpolation between graphs bounded degree and the general case (obtained by taking $d=n$) where we get the same bound as in \autoref{th:lbk3hit}.

\begin{figure}
 \centering
 \includegraphics[scale=1.5]{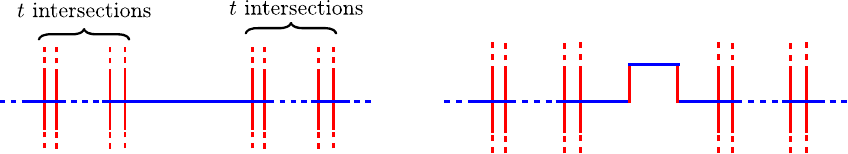}
 \caption{The gadget used in the proof of \autoref{th:lb-maxdeg}. An horizontal segment of a polygon before and after the addition of the crenellation in order to bound the number of crossings by $t+4$.}
 \label{fig:crenellationdegbis}
\end{figure}

\begin{restatable}{theorem}{apxthlbmaxdeg}
\label{th:lb-maxdeg}
    Under \ETH, the problems \TH and \OCT cannot be solved in time $2^{o\left (\sqrt{dn}\right )}$ on 
   \DEUXDIR{} graphs with maximum degree $d\geq 6$. 
\end{restatable}

\begin{proof}
    In this proof we construct the same \DEUXDIR graph $G$ from the proof of \autoref{th:lbk3hit}, but adapt it so that a side of a polygon is not crossed by more than $t$ segments. We do so by creating a graph $G'$ where crenelations, as depicted in \autoref{fig:crenellationdegbis}, are added on the top segments of the base rectangle of each variable gadget, as they are the only segments with unbounded degree in the construction.
    Observe that after the modification of a segment, each newly created horizontal segment has degree at most $t+4$, as a segment outside of a clause gadget intersects $2$ length $0$ segments. Taking $t=d-4$ ensures that the degree of the created segments are at most $d$ (a crenelation increases the degree by $2$ on each side), and the degree of the unmodified segments is at most $6\leq d$, so the maximum degree of $G$ is at most $d$.
    \begin{claim}
        $|G'|=\O\left(\frac{|G|^2}{d}\right)$.
    \end{claim}
    \begin{proof}
        In our construction of $G'$, for a segment intersected by $c$ segments, we add $\O(\frac cd)$ crenelations, each containing a constant number of segments. We have $c\leq |G|$ so the number of additional segments is $\O\left(\frac{|G|^2}{d}\right)$.
    \end{proof}

    Because our construction still satisfies the properties of \autoref{th:lbk3hit}, \autoref{cl:goodiffsat} holds.
     Therefore any algorithm solving \TH{} or \OCT{} in time $2^{o(\sqrt {dn'})}$ for input graphs of size $n'$ and maximum degree $d$ could be used to solve \THREESAT{} on formulae with $m$ clauses and in time $2^{o(m)}$, which would refute the \ETH.
\end{proof}

 We now provide in the next theorem a construction that is $K_{2,2}$-free, with maximum degree $6$, where segments do not cross (i.e., that is a \CONTACTDEUXDIR{} graph) and that applies also to \FVS{}, to the price of a weaker lower bound.

\begin{restatable}{theorem}{apxthlbfvs}
\label{th:lb-fvs}
    Under \ETH, the problems \TH, \OCT, and \FVS cannot be solved in time $2^{o\left (\sqrt{n}\right )}$ on $n$-vertex $K_{2,2}$-free 
   \CONTACTDEUXDIR{} graphs with maximum degree $6$. 
 \end{restatable}

Before starting the proof of \autoref{th:lb-fvs}, we need some intermediate results.
The \emph{incidence graph} of a \THREESAT{} formula is the (bipartite) graph whose vertices are clauses and variables and where edges connect variables to the clauses they appear in. The restriction of the \THREESAT{} problem to formulae with a planar incidence graph is called \PTHREESAT{}.

\begin{theorem}[\cite{lichtenstein1982planar}]\label{th:planar3sat}
There is no algorithm that solves  \PTHREESAT{} on a formula with $n$ variables and $m$ clauses in time $2^{o(\sqrt{n+m})}$, unless the \ETH fails.
\end{theorem}

 \begin{lemma}\label{lem:contactrect}
     There is an algorithm that, given a planar bipartite graph $G$, returns in polynomial time a representation of $G$ as a contact graph of rectangles, where two intersecting rectangles intersect on a non-zero length segment.
 \end{lemma}
 \begin{proof}
     Let $G$ be a planar bipartite graph. In linear time it can be represented as a \CONTACTDEUXDIR{} graph according to~\cite{CzyzowiczKU98BipartiteContact}, with the contact occurring only between segments of different directions. We will now transform this representation to obtain a representation of $G$ as a contact graph of rectangles. To do this,  we first prevent the contact of endpoints of two segments, by slightly extending one if such a contact point exists. This transformation done, we then shorten all the segments at both sides by a small $\epsilon$, and thicken them by the same amount.
     The segments are now interior disjoint rectangles, and two segments that were in contact are now two rectangles sharing an $\epsilon$-length segment on their border.
 \end{proof}
In order to avoid all-positive or all-negative clauses as in the proof of \autoref{th:lbk3hit}, we show below that those (if any) can be replaced without destroying planarity.

 \begin{lemma}\label{lem:planarclause}
     There is an algorithm that, given a planar \THREESAT{} formula with $n$ variables and $m$ clauses, returns in time polynomial in $n+m$ an equivalent planar \THREESAT{} formula with $n+m$ variables and $2m$ clauses where in addition no clause contains 3 positive literals or 3 negative literals.
 \end{lemma}

 \begin{proof}
     Let $\phi$ be the input formula with $n$ variables and $m$ clauses and let $G$ be its planar incidence graph. For a variable or clause $z$ of $\varphi$, we denote by $v(z)$ the corresponding vertex in $G$. If there is a clause $C$ of the form $x_a \lor x_b \lor x_c$ then we replace it by the two clauses $C_1 = x_a \lor x_b \lor \bar{y}$ and $C_2 = y \lor x_c$, where $y$ is a fresh variable. Clearly the obtained formula $\varphi'$ is equivalent and has $n+1$ variables and $m+1$ clauses. Also, the incidence graph $G'$ of $\varphi'$ can be obtained from $G$ by renaming $v(C)$ into $v(C_1)$ (for the clause $C_1$) and replacing the edge $v(x_c)v(C_1)$ by the path $v(C_1)v(y)v(C_2)v(x_c)$, where $v(C_2)$ and $v(y)$ are new vertices representing $C_2$ and $y$. As a subdivision of a planar graph, $G'$ is planar.
     A symmetric replacement can be done for clauses where all literals are negative. Each replacement decreases the number of all-positive or all-negative clauses so after $m$ replacement steps at most the obtained formula is free of such clauses (and equivalent to the initial one). Finding a clause to replace and doing so can be performed in polynomial time, hence so does the replacement of all all-positive and all-negative clauses, as claimed. 
 \end{proof}

\begin{proof}[Proof of \autoref{th:lb-fvs}.]
	In this proof we adapt our construction from \autoref{th:lbk3hit} so that the polygons intersect each other only at points of the form $z_C$, for $C$ a clause. toward this goal we consider an instance $\varphi$ of the problem \PTHREESAT{}; by \autoref{lem:planarclause} we may assume that $\varphi$ has no all-positive or all-negative clause; also we can easily simplify clauses of size one.
 
     The incidence graph $G_\varphi$ of $\varphi$ is a bipartite planar graph, so using \autoref{lem:contactrect} it can be represented as a contact graph of rectangles. Let us now describe how this representation is used to construct a \DEUXDIR{} instance of the cycle-hitting problems we consider. Initially, for each variable we create a 2-polygon that follows the associated rectangle (in the contact representation) and for each clause $C$ we choose a point $z_C$ of the plane located at the center of its rectangle. 
     For each clause, we want to add a few sides to the polygons of its ($2$ or $3$) variables so that these polygons intersect in $z_C$ only. As they do not intersect elsewhere, this ensures that the obtained graph is a \CONTACTDEUXDIR{}.
     We describe the construction for clauses with $3$ variables; that for two variables is a simpler version of it.
     Recall that for a clause $C$, the rectangles of its variables are in contact with the rectangle of $C$ and do not intersect mutually. For each variable in $C$, we add a non-zero length segment with one endpoint being $z_c$, vertical if the literal containing the variable is positive, horizontal otherwise. As the clause does not contain all-positive or all-negative literals, it is always possible to do so with the three segments intersecting in $z_C$ only.
     The choices of the sides of $z_c$ on which we put the segments is made so that the circular order of the segments around $z_C$ is the same as the circular order of the rectangles of the variables around that of $C$. 
     We claim that extending the polygons of the variables in order to connect each with the newly created segment without crossings can be done by adding a constant number of segments to each, as depicted in \autoref{fig:config-clause}.
     \begin{figure}
         \centering
         \includegraphics{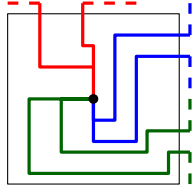}
         \caption{An example of extensions of three polygons toward a point $Z_C$ in order to encode a clause as described in the proof of \autoref{th:lb-fvs}.}
         \label{fig:config-clause}
     \end{figure}
     In order to finish the construction of the graph $G$, for each non-zero length segment, which is then a side of one of the polygons, we add in the portion of segment between the two intersections with the two adjacent sides of the concerned polygon a \emph{crenellation}, which is the simple gadget depicted in \autoref{fig:crenellation}. We make this crenellation with segments in contact, and small enough so that it does not intersect other segments. This corresponds to subdividing $4$ times each edge of the cycles corresponding to polygons. 

     \begin{figure}
         \centering
         \includegraphics[scale=.7]{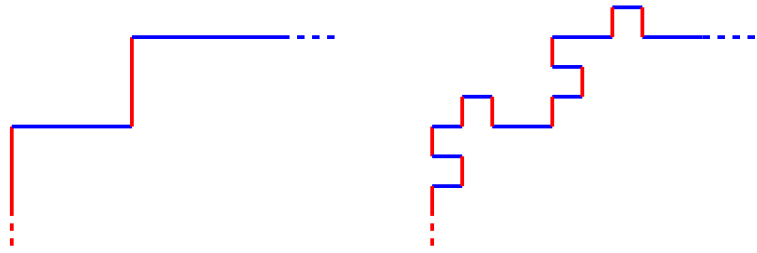}
         \caption{Portion of a polygon before and after the adding of the crenellation on each side of the polygon.}
         \label{fig:crenellation}
     \end{figure}

     \begin{claim}
        \label{cl:consecutive}
         A cycle in $G$ which is not a triangle contains $5$ consecutive sides of a same polygon. And so $G$ is $K_{2,2}$ free.
     \end{claim}
     \begin{proof}
         Let $Q$ a cycle of size at least $4$. A $z_C$ point is contained in only $3$ segments and the segments containing a $z_C'$ point with $C'\neq C$ are at distance at least $4$ because of the crenellation. This ensures there is in $Q$ a segment $s$ which does not contain a $z_C$ point, so which is a part of a polygon and more precisely of a crenellation. But if a cycle hits a segment from a crenellation that does not contain a $z_C$ point, all the paths that the cycle can follow over the crenellation go from one of its side to another, crossing the $5$ adjacent non-zero length segments of the crenellation. Observe that $K_{2,2}$ is a cycle of size $4$, and we just proved that a cycle that is a subgraph of $G$ is either of size $3$ or at least $5$, so $G$ is $K_{2,2}$ free.\cqed
     \end{proof}
     Observe that the graph $G$ obtained by the above construction satisfies the same properties as the graph constructed in the proof of \autoref{th:lbk3hit}.
     Additionally,
     \begin{enumerate}
         \item $G$ is a \CONTACTDEUXDIR{} (as argued above);
         \item $G$ is $K_{2,2}$-free, as proved in \autoref{cl:consecutive}.
         \item the polygons corresponding to two variables intersect only in a point of the form $z_C$, for $C$ a clause where they both appear.
         \item $G$ have degree at most $6$, where this value could be attained by a segment in contact with $z_C$, see for example \autoref{fig:config-clause}.
     \end{enumerate}
     Now that the graph is defined, let us show how it relates to the cycle-hitting problems.
     As in the previous proof, we denote $k_i$ the number such that the variable gadget for the variable $x_i$ is a $k_i$-polygon and $k=\sum_{i=1}^nk_i$. 
     \begin{claim}
         For a set $S\subseteq V(G)$ of size at most $k$, the three following properties are equivalent:
         \begin{enumerate}
            \item \label{e:fvs}$S$ is a feedback vertex set,
            \item \label{e:oct}$S$ is a odd cycle transversal,
            \item \label{e:th}$S$ is a triangle hitting set.
         \end{enumerate}
    \end{claim}
    \begin{proof}
        We trivially have \eqref{e:fvs}$\Rightarrow$\eqref{e:oct}$\Rightarrow$\eqref{e:th} so it is enough to prove \eqref{e:th}$\Rightarrow$\eqref{e:fvs}. Let $S$ be a triangle hitting set of $G$ of size at most $k$. Hitting all the triangles of the polygons of $G$ requires already $k$ vertices, so $S$ has size exactly $k$ and so is a good triangle hitting set as defined in the proof of \autoref{th:lbk3hit}.
        \autoref{cl:consecutive} assures that any cycle of length greater than $3$ contains two adjacent non-zero length segments from a same polygon, and so with different directions. 
        But because $S$ is a good triangle hitting set, we know that one of those two segments will be part of $S$. This proves that $S$ not only hits all the triangles, but it is a set hitting every cycle of $G$.\cqed
    \end{proof}

    Because our construction still verifies the properties of the construction done in the \autoref{th:lbk3hit}, the \autoref{cl:goodiffsat} is still verified.
     So an algorithm solving \FVS, \TH or \OCT in time $2^{o(\sqrt n)}$ would also solve \PTHREESAT with the same complexity, which according to \autoref{th:planar3sat} contradicts the \ETH.
\end{proof}

\section{Positive results for \trh{}}
\label{sec:sublin}

In this section, we provide a subexponential parameterized algorithm
for \TH{} in $K_{t,t}$-free \dDIR graphs, $K_{t,t}$-free string
graphs, and \CONTACTSEG graphs. These results follow from a more
general result on graph classes with strongly sublinear separators
described in the first subsection.

\subsection{A subexponential FPT algorithm in classes with sublinear separators}

The goal of this section is to prove the following result (restated
below in a more accurate form as \autoref{th:lineigh-subexp}).
\begin{theorem}
If $\G$ is a hereditary graph class that has strongly sublinear separators, then there is a subexponential time FPT algorithm for \TH in $\G$.
\end{theorem}

The proof of this result easily follows by combining known
results. We include because one of its consequences (described below)
is that $K_{t,t}$-free string graphs admit subexponential parameterized algorithm
for \TH{}, which is a positive counterpart to the lower-bound in \autoref{th:lb-fvs}.

\begin{definition}[\cite{dvorak2016strongly}]
We say that a hereditary graph class $\mathcal{G}$ has \emph{strongly sublinear separators} if there is a function $f\colon \NN\to\RR$ with $f(n) = \O(n^\delta)$ for some $\delta\in [0,1)$ such that every $G\in \mathcal{G}$ has a balanced separator of order at most $f(|G|)$.
\end{definition}

\begin{definition}\label{def:neighc}
If for a graph class $\mathcal{G}$ there is a constant $c$ such that for every $G \in \mathcal{G}$ and every $X\subseteq V(G)$,
$
|\{N(v) \cap X : v \in V(G)\}| \leq c |X|,
$
then we say that $\mathcal{G}$ has \emph{linear neighborhood complexity} with \emph{ratio}~$c$.
\end{definition}

Strongly sublinear separators and neighborhood complexity are linked by the following results.
To avoid definitions that are otherwise irrelevant to this paper, we skip the definitions of polynomial and bounded expansion used in the next two results.

\begin{theorem}[\cite{dvorak2016strongly}]
\label{th:polyexp}
Every graph class with strongly sublinear separators has polynomial expansion.
\end{theorem}

\begin{lemma}[\cite{reidl2019characterising}]
\label{lem:reidl}
    If a graph class has bounded expansion, then it has linear neighborhood complexity.
\end{lemma}

We will also need the following connection to treewidth.
\begin{theorem}[\cite{DVORAK2019137}]
\label{th:septw}
If a hereditary graph class $\mathcal{C}$ has strongly sublinear separators with function $f(n) = \beta n^\delta$, then the treewidth of any $n$-vertex graph in $\mathcal{G}$ is at most $15 \beta n^\delta$. 
\end{theorem}

The proof of the following result follows the same steps as the proof in~\cite{lokSODA22} for disk graphs, but as our statement is more general,
we prefer to include it.

\begin{theorem}\label{th:lineigh-subexp}
If $\G$ is a hereditary graph class that has strongly sublinear separators, then there is a subexponential time algorithm for \TH in $\G$.\\
More precisely, if the size of the separators in $\mathcal{G}$ is
$f(n) = \beta n^{\delta}$ for some $\delta<1, \beta>0$ and the ratio of neighborhood complexity is $c$, then the algorithm runs on $n$-vertex graphs of $\G$ with parameter $k$ in time $
2^{\O(\beta \gamma^\delta c^\delta k^{2\delta/(1+\delta)}\log k)} n^{\O(1)}.
$%
for some universal constant $\gamma$.
\end{theorem}

Notice that by \autoref{th:polyexp} and \autoref{lem:reidl}, the ratio $c$ as in the above statement is always defined for graph classes with strongly sublinear separators.

\begin{proof}[Proof of \autoref{th:lineigh-subexp}.]
We prove the statement for the (more general) \wtrh problem.
This is the weighted version of \TH where the vertices of the input graph have positive weights and the aim is to decide if there is a triangle hitting set whose sum of weights is at most the input parameter.
    Let $\alpha = \frac{1-\delta}{1+\delta}$.  We apply \autoref{cor:bothbranchings} with $p = k^{\alpha}$, and thus in $2^{\O(k^{1-\alpha} \log k)} n^{\O(1)}$ time we generate the $2^{\O(k^{1-\alpha} \log k)}$ instances $(G',M,k')$ as described in the aforementioned corollary.
    Let us now describe how we solve these instances.

    By definition of $c$ the vertices of $G'-M$ can be partitioned into at most $c|M|$ classes $V_1, \dots, V_r$ such that two vertices in the same class have the same neighborhood in $M$.
    For every $i \in \intv{1}{r}$ we delete all but one vertex of $V_i$, that we call $v_i$. Let $G''$ be the obtained graph. Let $w$ be the weight function of the input instance.
    In order to keep track of the deleted vertices we define a new weight function $w''\colon V(G'')\to \NN$ as $w''(v_i) = \sum_{v\in V_i} w(v)$ for every $i \in \intv{1}{r}$ and $w''(v) = w(v)$ for every $v \in M$. The produced instance is $(G'', w'', k')$; we call it a \emph{translation} of the instance $(G',M,k')$. Observe that it can be computed in polynomial time, and that it is an instance of \wtrh{}.

    \begin{claim}
    \label{cl:fusion}
    $(G',w, k')$ is a yes-instance or \TH{} $\Leftrightarrow$ $(G'', w'', k')$ is a yes-instance of \textsc{Weighted} \TH{}.
    \end{claim}
    \begin{proof}
    Observe first the property $(P1)$ that, as for any $v \in M$, $N(v)\setminus M$ is an independent set, and $M$ is a triangle hitting, any triangle $\Delta$ of $G'$ (or $G''$) has $|\Delta \cap M| \ge 2$.
    
    Direction ``$\Leftarrow$''. Let $S''$ be a solution of $(G'', w'', k')$. Observe that every vertex of $S''$ either belongs to $M$ or is of the form $v_i$ for some $i\in \intv{1}{r}$ (using the same notation as above). Let $I\subseteq \intv{1}{r}$ be the set of integers such that $S''\setminus M = \{v_i : i\in I\}$.
   We define $S' = (S'' \cap M) \cup \bigcup_{i\in I} V_i$.
    By the definition of $k'$ and $w$, we have $\sum_{v\in S'} w(v) \leq k'$. In order to show that $S'$ is a triangle hitting set of $G'$, we assume toward a contradiction that $G' - S'$ has a triangle 
    $\Delta$. If $|\Delta \cap M|=3$ then $\Delta$ would be a triangle in $G'' - S''$, a contradiction, so by property $(P1)$ we have $|\Delta \cap M|=2$, and thus let $v_i$ be the representative in $G''$ of the vertex in $\Delta \setminus M$. By construction of $S'$, this implies that $v_i \notin S''$, and even that $\Delta \cap S'' = \emptyset$, a contradiction.

    Direction ``$\Rightarrow$''. Let $S'$ be a subset-minimal solution of $(G',w, k')$. 
    Observe that for every $i \in \intv{1}{r}$, either $V_i \subseteq S'$ or $V_i \cap S' = \emptyset$. Indeed, if $v
    \in V_i$ belongs to $S'$, by minimality of this set there is a triangle $\Delta$ containing $v$ in $G'-(S'\setminus \{v\})$, 
    and by property $(P1)$ it has its two other vertices in $M$. By definition of $V_i$, every $u\in V_i$ is neighbor of these two vertices. So, as $S'$ is a triangle hitting set, $u\in S'$.
    Let $I = \{i \in \intv{1}{r}: S\cap V_i \neq \emptyset\}$.
    Let $S'' = (S'\cap M) \cup \{v_i : i\in I\}$, and notice that this set has weight at most $k'$. So it remains to show that $S''$ is a triangle hitting set of $G''$. Suppose toward a contradiction that $G''-S''$ has a triangle $\Delta$.
    As previously, if $|\Delta \cap M|=3$ we have a contradiction and thus $|\Delta \cap M|=2$, and let $v_i$ be the vertex in $\Delta \setminus M$. By definition of $S''$, $v_i \notin S'$ and thus $\Delta$ is a triangle of $G'-S'$, a contradiction.
    \cqed
   \end{proof}

    In remains now to solve such an instance $(G'',w'',k')$ of \wtrh{}. Let $t = \tw(G'')$.
    We first compute in $2^{\O(t)}|G'|$ steps a tree decomposition of width $\O(t)$ (for instance using the approximation algorithm of Korhonen~\cite{korhonen2022single}). We then solve the problem in $2^{\O(t)}|G'|^{\O(1)}$ steps via a standard dynamic programming algorithm on the tree decomposition by observing that every triangle lies in some bag (as observed by~\cite{cygan2017hitting} for the general problem of hitting cliques).

        According to \autoref{cor:bothbranchings}, $|M| \leq 11k^{1+\alpha}$ so $|G''| \leq 11(c+1)k^{1+\alpha}$.
    Recall that $\mathcal{G}$ is hereditary and $G''$ is an induced subgraph of $G'$ that is an induced subgraph of $G$, so $G'' \in \mathcal{G}$.
    As $\mathcal{G}$ has strongly sublinear separators with bound $f$, by \autoref{th:septw} the treewidth $t$ of $G''$ is at most $15 \beta |G''|^\delta$, so $t \leq 15 \beta \gamma^\delta c^\delta k^{(1+\alpha)\delta}$ for some constant $\gamma>0$.
    The number of instances generated by \autoref{cor:bothbranchings} is at most $2^{\O(k^{1-\alpha} \log k)}$ so overall the running time is
    \[
        2^{\O(\beta \gamma^\delta c^\delta  k^{(1+\alpha)\delta} + k^{1-\alpha} \log k)}  n^{\O(1)}%
        \leq %
        2^{\O(\beta \gamma^\delta c^\delta k^{2\delta/(1+\delta)}\log k)}  n^{\O(1)},
    \] as claimed.
\end{proof}

\subsection{Application  to \texorpdfstring{$K_{t,t}$}{Ktt}-free \texorpdfstring{\dDIR}{d-DIR}, \texorpdfstring{$K_{t,t}$}{Ktt}-free string graphs and \CONTACTSEG graphs}
\label{sec:basic-lineigh}

 \autoref{th:lineigh-subexp} from the previous section can be directly applied to $K_{t,t}$-free string graphs. Indeed, a straightforward consequence of the following results of Lee is that such graphs have strongly sublinear separators.
\begin{theorem}[\cite{lee2016separators}]\label{thm:string-sep}
    $m$-edge string graphs have balanced separators of size $\O(\sqrt{m})$.
\end{theorem}
\begin{corollary}\label{cor:sep-m-tw}
    Any $m$-edged string graph has treewidth $\O(\sqrt{m})$.
\end{corollary}
This directly follows from  \autoref{th:septw} and \autoref{thm:string-sep}.

\begin{theorem}[\cite{lee2016separators}]\label{th:neighbip}
There is a constant $c$ such that for every integer $t>0$, every $K_{t,t}$-free string graph on $n$ vertices has at most $c  t(\log t)  \cdot n$ edges.
\end{theorem}

\begin{corollary}\label{cor:ktstr}
$K_{t,t}$-free string graphs on $n$ vertices have strongly sublinear separators of order $\O(\sqrt{nt\log t})$, and by \autoref{th:septw} they have treewidth at most $\O(\sqrt{nt\log t})$.
\end{corollary}

\autoref{cor:ktstr} and \autoref{th:lineigh-subexp} applied with $\delta=1/2$, $\beta = \O(\sqrt{t\log t})$ and for $c$ any bound on the neighborhood complexity (that is bounded from above by a function $\delta$, according to \autoref{th:polyexp} and \autoref{lem:reidl}) lead to the following result.

\begin{restatable}{theorem}{apxthtrhstr}\label{th:trhstr}
    For every integer $t$ there is a constant $c_t>0$ such that the following holds.
    There is an algorithm for \TH in string graphs that runs on $n$-vertex instances with parameter $k$ in time
    $
    2^{c_t k^{2/3} \log k} n^{\O(1)},   
    $
    where $t$ is the minimum integer such that $G$ is $K_{t,t}$-free, even when neither $t$ nor a string representation are given as input.
\end{restatable}

In \autoref{ssec:neighcplx} we provide the following explicit bounds on $c_t$ in $K_{t,t}$-free \dDIR and \CONTACTSEG graphs.

\begin{restatable}{lemma}{firstingredient}
\label{lem:first-ingredient}
For every integer $t>0$, the class of 
$K_{t,t}$-free
\dDIR{} graphs has linear neighborhood complexity with ratio $\O(d t^3 \log t)$.
\end{restatable}

\begin{restatable}{lemma}{firstingredientcontact}
\label{lem:first-ingredient-contact}
The class of \CONTACTSEG graphs has linear neighborhood complexity.
\end{restatable}

\begin{lemma}\label{lem:kt-vs-ktt}
    A \CONTACTSEG graph with clique number $\omega$ does not contain a $K_{12\omega,12\omega}$ as a subgraph.
\end{lemma}
\begin{proof}
    Towards a contradiction, assume there is a \CONTACTSEG graph $H$ with clique number $\omega$ containing a $K_{12\omega,12\omega}$-subgraph, with vertex sets $A=\{a_1,\ldots,a_{12\omega}\}$ and $B=\{b_1,\ldots,b_{12\omega}\}$. We denote by $\alpha_A$ and $\alpha_B$ the sizes of a maximum independent set in $H[A]$ and $H[B]$, respectively. As these graphs are not complete, we have $\alpha_A\ge 2$ and $\alpha_B\ge 2$.
    
    Every triangle-free \CONTACTSEG graph (as it only contains trivial contact points) is planar (see Lemma 2 in~\cite{deniz2018contact}),   
    so in particular $K_{3,3}$ is not a \CONTACTSEG graph.
    Therefore there is no induced $K_{3,3}$ in $H$, which implies that $\alpha_A\le 2$ or $\alpha_B\le 2$.
    In the following we assume that $\alpha_A = 2$ and $\alpha_B\ge 2$. Let $b_1,b_2\in B$ be non-adjacent.

    We consider a maximal subset $A'\subseteq A$ whose segments touch $b_1$ on the same side, and touch $b_2$ on the same side. This set has size at least $2\omega$, as there are at most $4\omega$ segments in $A$ touching $b_1$ or $b_2$ on their endpoints, and as the other segments can be partitioned into four sets according to the sides of $b_1$ and $b_2$ they reach these segments. 

    As $H$ is a contact-segment graph, the segments in $A'$ can be ordered $a'_1,\ldots,a'_{2\omega},\ldots$ so that $a'_i$ and $a'_j$, with $i<j$ do not intersect if $j-i\ge \omega-1$ (otherwise, all the segments in-between and $b_1$ would intersect and form a $(\omega+1)$-clique). Thus $a'_1$, $a'_{\omega}$, and $a'_{2\omega}$ form an independent set, a contradiction to $\alpha_A=2$.
    \end{proof}

This leads to the following result.

\begin{restatable}{corollary}{apxcorth}
\label{cor:th}
    There is an algorithm that solves \TH in time
    \begin{itemize}
        \item $2^{\O(k^{2/3}(\log k) \sqrt d t^2 \log t )} n^{\O(1)}$ in $K_{t,t}$-free \dDIR{} graphs,
       \item $2^{\O(k^{3/4} \log k)} n^{\O(1)}$ in \CONTACTSEG graphs,
    \end{itemize} 
    even when no representation is given.    
\end{restatable}

\begin{proof}
    For the first result, \autoref{cor:ktstr}, together with \autoref{th:lineigh-subexp} with $\delta=1/2$, $\beta = \sqrt{t\log t}$, $c=\O(d t^3\log t)$ (\autoref{lem:first-ingredient}), implies the claimed running time.

    For the second result, given an instance $(G,k)$ of \TH in \CONTACTSEG graphs, we consider $p=k^{\alpha}$ (with $\alpha$ to be chosen), 
    and proceed as in the proof of \autoref{th:lineigh-subexp} by first
    generating the $2^{\O(k^{1-\alpha}log(k))}$
    instances $(G'', k, w)$ of \wtrh.
    As \CONTACTSEG have linear neighborhood complexity, we have $|G''|=O(|M|)$.
    Moreover, as $G'$ is the result of applying \autoref{cor:bothbranchings}, we have $\omega(G') \le p$, and \autoref{lem:kt-vs-ktt} implies that $G'$ is $K_{t,t}$-free for $t=\O(p)$. By \autoref{cor:ktstr}, this implies that $\tw(G'')=\O(\sqrt{|M|t\log(t)})=\O(p\sqrt{k\log(k))}$.
    By choosing $\alpha=1/4$, and solving the instances in the same way than in \autoref{th:lineigh-subexp}, we obtain the wanted result.
\end{proof}

\subsection{Neighborhood complexity}
\label{ssec:neighcplx}

In this section we prove the aforementioned result about the
neighborhood complexity of $K_{t,t}$-free \dDIR{} graphs.

\firstingredient*

\begin{proof}
	Let $G$ be a $K_{t,t}$-free \dDIR{} graph, and fix $M\subseteq V(G)$. Observe that 
 $|\{N(v)\cap M: v \in V(G)\}| \le |M|+|\{N(v)\cap M: v \in V(G) \setminus M\}|$, and thus we only have to prove that
 $|\{N(v)\cap M: v \in V(G) \setminus M\}|=\O(|M| \cdot d t^3 \log t)$.
 
     Whenever two vertices of $G-M$ have the same neighborhood in $M$, we delete one of them. For simplicity we also delete all vertices with no neighbors in $M$ and call $G'$ the obtained graph. We now focus on bounding the order of $G'$; observe that this would straightforwardly translate to a  bound on the number of neighborhoods in $M$ of vertices of $G-M$.
	
	For a \dDIR{} representation of $G'$, we consider that one of its direction is the horizontal. So in the following we say that a vertex of $G'$ is \emph{horizontal} if its segment in the representation is. For convenience the vertices with $0$-length segments are considered as horizontal vertices.

    Let $M_H$ denote the horizontal vertices of $M$ and $M_C=M\setminus M_H$. Moreover we set $A=V(G')\setminus M$ and define in the same way $A_H$. As the horizontal direction was chosen arbitrarily, a bound of the size of $A_H$ multiplied by $d$ would be a bound to the size of $A$.
    Towards this goal we further partition $A_H$ into two subsets: $A_{H_1}$ and $A_{H_2}$ with $A_{H_1}=\{v\in A_H : N(v)\cap M_H\neq \emptyset\}$ and $A_{H_2}=A_H\setminus A_{H_1}$.

        \begin{claim}\label{cl:ah1}
          $|A_{H_1}| = \O(|M| \cdot t^2  \log t)$.
        \end{claim}
	\begin{proof}[Proof of \autoref{cl:ah1}]
		
	We will bound $|A_{H_1}|$ by using special points of the plane contained in the horizontal segments of $M$. Let $P_H$ be the set of the endpoints of the $M_H$ segments and the points of intersection of a $M_H$ segment with a $M_C$ segment.


 We denote by $m_M$ the number of edges of the graph $G'[M]$. The size of $P_H$ is at most $2|M_H| + m_M$ as intersections between segments correspond to edges in $G'[M]$, whose number is bounded by $\O(|M| \cdot t \log t)$ in \autoref{th:neighbip}.


		 There are at most $|P_H|$ segments of $A_{H_1}$ which does not intersect any point of $P_H$: those are the ones strictly included in a subsection of the horizontal segments of $M$ split at the $P_H$ points. For the remaining points, which each intersect at least one point of $P_H$, there cannot be more than $2t|P_H|$ of them as otherwise by the pigeonhole principle there would be a point of $P_{H}$ contained in at least $2t$ segments. But then $K_{2t}$ would be a subgraph of $G$, and so $K_{t,t}$, contradicting our hypothesis. So overall we get $|A_{H_1}| \leq (2t+1)|P_H|\leq (2t+1)(2|M_H| + m_M)=\O( t^2  (\log t) |M|)$.\cqed
	\end{proof}
	We now want to bound the size of $A_{H_2}$, the horizontal vertices of $G-M$ which do not intersect a horizontal segment of $M$.

     \begin{claim}\label{cl:ah2}
         $|A_{H_2}| = \O(t^3 (\log t) |M|)$.
    \end{claim}
	\begin{proof}[Proof of \autoref{cl:ah2}.]

	Let us first remark that we can safely ignore the horizontal vertices of $M$ as they have no neighbors in $A_{H_2}$. To simplify the proof, we add a perturbation to the segments of the representation of $G'$, keeping the property that this is a \dDIR representation of $G'$ but such that an intersection involves at most $2$ different directions, and such that all the endpoints and the intersections between segments of $M_C$ have distinct ordinate. 
 For achieving such perturbation, we extend the length of every segment of the representation of $G'$ by a small amount. 
 Moreover for each line that supports at least one segment, we move slightly the whole line to a near parallel line (ensuring the associated intersection graph remains the same). Once this perturbation done, we define a set of special points in the plane, $P_C$, which are the endpoints of the $M_C$ segments and theirs intersections points with each other. The perturbation ensure that the $P_C$ points have distinct ordinates.
 Let denote $p_1,\dots,p_k$ the points of $P_C$ ordered from top to bottom, and for $1\leq i \leq k$ we denote $L_i$ the horizontal line containing the point $p_i$. We now partition $A_{H_2}$ according to their position relatively to the $(L_i)_i$ lines: for $1\leq i \leq k$ we denote by $Z_i$ the set of segments of $A_{H_2}$ whose neighborhood in $M$ can be achieved by an horizontal segment above $L_{i+1}$ but not by a segment on or above the line $L_i$. (see \autoref{fig:dDIRneigh})
 \begin{figure}
        \centering
       	\includegraphics[]{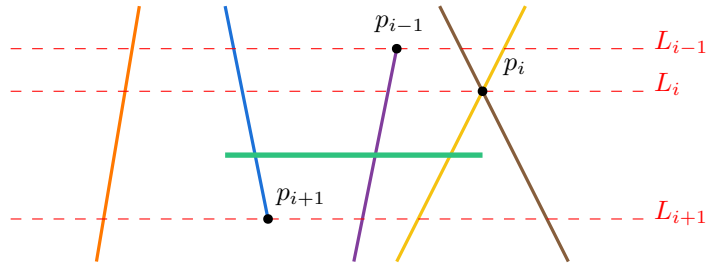}
   	\caption{Example for the construction of the family $(Z_i)_i$ in the proof of \autoref{cl:ah2}. The green segment would be part of the set $Z_i$.}
    \label{fig:dDIRneigh}
   \end{figure}

    Observe that:
        \begin{itemize}
            \item $(Z_i)_i$ is a partition of $A_{H_2}$, as a segment that can be represented on one line $L_i$ will have a direction, either up or down, where it can be moved by a small amount without changing its neighborhood.
            \item In the zone between lines $L_i$ and $L_{i+1}$, there would be no endpoints of segments of $M_C$, and segments of $M_C$ do not cross each other, which means they will have a constant left to right order. We will denote $S_i=\{s_1,\ldots, s_{l}\}$ those segments enumerated from left to right .
            \item For a segment of $Z_i$, its neighborhood is an interval in the sequence $(s_i)_i$.
        \end{itemize}
        
         We now want to prove that there exist two segments in $S_i$ such that every segment in $Z_i$ intersects at least one of them. We differentiate three cases depending of the point $p_i$:
        \begin{itemize}
            \item if $p_i$ is an intersection of segments of different direction, the perturbation ensure that at most two directions are involved. So there is two segments in $S_i$, each containing the point $p_i$ and one from each involved direction, such that a segment in $Z_i$ cross exactly one of those two segments, as otherwise a representation of this segment with this neighborhood in $M$ could be raised above the line $L_i$.
            \item if $p_i$ is a top endpoint of a segment $s$, then each segment $z$ of $Z_i$ would intersect this segment $s$ as otherwise the representation of $z$ below $L_i$ could be raised up above $L_i$.
            \item if $p_i$ is a bottom endpoint of a segment $s$, if this endpoint was at the left of each segment in $S$ then all the segment that could be done below $L_i$ could still be represented with the same neighborhood above $L_i$ and so $Z_i=\emptyset$, the same is true if the endpoint is on the right of each segment of $S_i$. So we can assume that the endpoint is between two segments of $S$. Then a segment of $Z_i$ will necessarily intersect these two segments as otherwise the representation of $z$ below $L_i$ could be raised up above $L_i$.
        \end{itemize}


    Now than we have this property, let $s_j$ be the segment of $S_i$ that is crossed by at least half of $Z_i$.
    There are at most $t^2$ segments that do not cross the segment $s_{j-t}$, nor $s_{j+t}$ as this is the number of intervals in $[j-t,\ldots,j,\ldots,j+t]$ that contain $j$.
    There are at most $t-1$ segments intersecting $s_{j-t}$ (if it exists), as otherwise, there would be a $K_{t,t}$. Similarly there are at most $t-1$ segments intersecting $s_{j+t}$. So $|Z_i|\leq 2(t^2+2t-2)=\O(t^2)$.
    
    As the $(Z_i)_i$ was a partition of $A_{H_2}$, and the number of special points $|P_C|$ is at most $2|M_C|+m_M=\O(|M|t \log t)$, we obtain
    $|A_{H_2}|=\O(|M| t^3  \log t)$.\cqed
    \end{proof}

To conclude, 
\[|\{N(v)\cap M : v \in V(G) \setminus M\}|\le 1 + d(|A_{H_1}| + |A_{H_2}|) = \O(|M| \cdot d  t^3\log t). \]
\end{proof}

\firstingredientcontact*

\begin{proof}
    Let $G$ be a \CONTACTSEG graph, and fix $M\subseteq V(G)$. As in the previous lemma, let us only prove that $|\{N(v)\cap M : v \in V(G) \setminus M\}|=\O(|M|)$.
     Whenever two vertices of $G-M$ have the same neighborhood in $M$, we delete one of them. 
     We also delete all vertices whose neighborhood in $M$ can be obtained with a zero-length segment: this kind of neighborhood either contains at most one vertex of $M$ (there are $|M|+1$ such neighborhoods) or is a set of segments in contact on the same point. This point has to be an endpoint of at least one segment of $M$, so there are at most $2|M|$ such neighborhoods. 
     We call $G'$ the graph obtained after deleting those vertices.
     We split the set $A=V(G')\setminus M$ in two parts such that $A_1$ is the set of segments of $A$ whose interior is in contact with an endpoint of a segment of $M$, and $A_2=A\setminus A_1$. Observe that an endpoint of a segment can be in contact with only one interior of segments without creating a crossing, so we have $|A_1|\leq 2|M|$. Now it remains to bound $A_2$.
     \begin{claim}
        \label{cl:contactA2}
         $|A_2|\leq 24|M|$.
     \end{claim}
     \begin{proof}
        Let us define an embedding in the plane of a planar graph $H$ such that $|H|\leq 8|M|$ and $|E(H)|\geq |A_2|$. The embedding is defined from a \CONTACTSEG representation of $G' - A_2$, as depicted in \autoref{fig:contact-neigh}.
        \begin{figure}
            \centering
            \includegraphics[scale=.7]{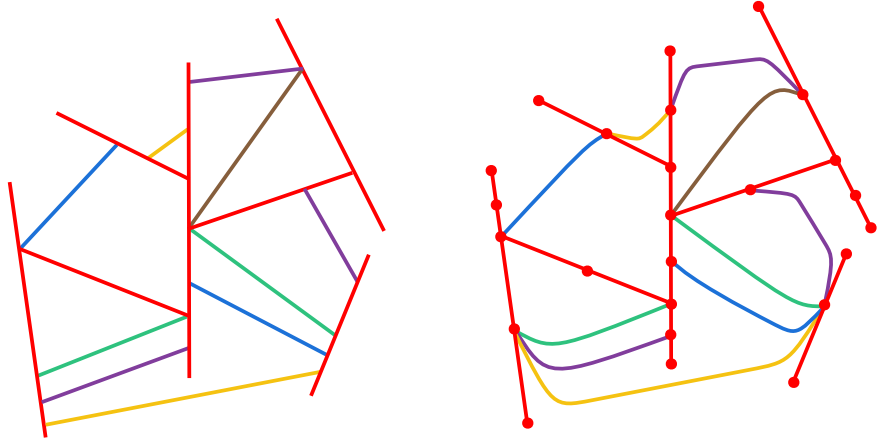}
            \caption{Example of the construction used in the proof of \autoref{cl:contactA2}. The \CONTACTSEG representation of $G'$ is on the left, with the segments of $M$ in red. On the right is the the constructed planar graph $H$, with the edges not in red being the representation of the segments of $A_2$.}
            \label{fig:contact-neigh}
        \end{figure}
        Firstly we add one vertex on every position of the endpoints of the segments of $M$, and denote $V_1$ this set of vertices, where $|V_1| \le 2|M|$.
        Then, we add edges between vertices of $V_1$ as follows: for each segment $v$ of $M$, we add a path starting at an endpoint of $v$ and then (following their order along $v$) all the endpoints of segments of $M$ on segment $v$, until reaching the other endpoint of $v$. We draw the edges of this path following the segment representation of $v$.
        We denote this set of edges by $E_1$.
        Observe that at this step the graph is planar, implying that $|E_1|\le 6|M|$ by the Euler's formula. 
        
        The next step of the construction is to subdivide every edge $e \in E_1$ by adding a vertex $v_e$ on the center $e$. We denote $V_2$ this set of vertices. Informally, we associated a vertex of $H$ to every section of segment between endpoints of $M$. Observe that $|V_1|+|V_2| \le 8|M|$.
        
        Let us now associate to each $v \in A_2$ a new edge $e(v)$ in $H$, and also explain how we can draw these new edges without crossings.
        Let $v \in A_2$. 
        Observe that as $v$ is not in contact with $M$ in its interior, its neighborhood in $M$ is entirely decided by its two endpoints, and both of them are part of a segment of $M$ as otherwise the neighborhood in $M$ of $v$ could be achieved with a zero-length segment, and so $v\notin V(G')$.
        Let $s$ be the segment representing $v$ and $\{p_1(s),p_2(s)\}$ be the endpoints of $s$.  Let slightly shortening $s$ to obtain $\tilde{s}$, and for $i \in \{1,2\}$ let $p_i(\tilde{s})$ be the endpoint of $\tilde{s}$ corresponding to $p_i(s)$.
        We now associate a vertex $v_i \in V(H)$ to each $p_i(s)$ and we will define $e(v)= \{v_1,v_2\}$. 
        Moreover, to draw $e(v)$, we will also define three segments $s_1(v),s_2(v)$ and $s_3(v)$, and draw $e(v)$ as $s_1(v) \cup s_2(v) \cup s_3(v)$. Firstly, we define $s_3(v)=\tilde{s}$. Then, let us distinguish two cases:
\begin{enumerate}
            \item If $p_i(s)$ is also an endpoint $p$ of a segment of $M$, then $v_i$ is the vertex of $H$ corresponding to the endpoint of this segment in $M$, and $s_i(v)=\{p,p_i(\tilde{s})\}$ 
            \item Otherwise, $p_i(s)$ is in the interior of exactly one segment $s$ of $M$, and more precisely inside an edge $e \in E_1$.
            Then, we define $v_i=v_e$. Let $p(v_e)$ be the point associated to $v_e$. We define $s_i(v)=\{p(v_e),p_i(\tilde{s})\}$.
\end{enumerate}
This concludes the definition of $H$. 
As required, we get that $H$ is planar, $|H|\leq 8|M|$, $|E(H)|\geq |A_2|$. This last property gives $|E(H)|\leq 3|H|$ by the Euler's formula, and so $|A_2|\leq 24|M|$, which is the wanted result.\cqed
     \end{proof}
     In total,  $|\{N(v)\cap M: v \in V(G) \setminus M\}| \le 1+|M|+|A_1|+|A_2|=\O(|M|)$.
\end{proof}

\section{Discussion}
In this paper we gave subexponential FPT algorithms for cycle-hitting problems in intersection graphs. A general goal is to characterize the geometric graph classes that admit subexponential FPT algorithms for the problems we considered.  
In particular, an interesting open problem is whether \FVS admits a subexponential parameterized algorithm in \DEUXDIR{} graphs.

\bibliography{biblio}

\end{document}